\documentclass{article}
\usepackage[utf8]{inputenc} 
\usepackage[T1]{fontenc}
\usepackage{microtype}
\usepackage{graphicx}
\usepackage{subfigure}
\usepackage{booktabs} 
\usepackage{amsmath, amsfonts, bm, amsthm}
\usepackage{enumerate}
\usepackage{hyperref}
\usepackage{epstopdf}
\usepackage{changepage}
\usepackage{algorithm,algorithmic,url}

\theoremstyle{definition}
\newtheorem{theorem}{Theorem}[section]
\newtheorem{lemma}{Lemma}[section]
\newtheorem{proposition}{Proposition}[section]

\newtheorem{remark}{Remark}[section]

\usepackage{geometry}
\title{Robust Decoding  from  Binary    Measurements with  Cardinality Constraint    Least Squares}
\author{Zhao Ding\thanks{School of Mathematics and Statistics, Wuhan University, Wuhan 430072, P.R. China. (zd1998@whu.edu.cn)}
\and  Junjun Huang \thanks{School of Mathematics and Statistics, Wuhan University, Wuhan 430072, P.R. China. (hjj$\_$wd@whu.edu.cn)}
\and Yuling Jiao\thanks{School of Mathematics and Statistics, and
Hubei Key Laboratory of Computational Science, Wuhan University, Wuhan 430072, P.R. China. (yulingjiaomath@whu.edu.cn)}
\and Xiliang Lu\thanks{School of Mathematics and Statistics, and
Hubei Key Laboratory of Computational Science, Wuhan University, Wuhan 430072, P.R. China. (xllv.math@whu.edu.cn)}
\and Zhijian Yang \thanks{School of Mathematics and Statistics, and
Hubei Key Laboratory of Computational Science, Wuhan University, Wuhan 430072, P.R. China. (zjyang.math@whu.edu.cn)}
}
\begin{document}

\maketitle

\begin{abstract}
The main goal of 1-bit compressive sampling is to decode
   $n$ dimensional signals with sparsity level $s$ from $m$  binary  measurements. This is a challenging task  due to the presence of  nonlinearity, noises  and sign flips.  In this paper,
    the cardinality constraint  least square  is proposed as
 a desired decoder.
 We prove that, up to a constant $c$, with high probability, the proposed decoder achieves a   minimax estimation error  as long as  $m \geq \mathcal{O}( s\log n)$.
 Computationally,
we utilize a generalized Newton   algorithm  (GNA) to solve  the  cardinality constraint  minimization problem with the  cost of solving a  least squares problem  with small size at each iteration.
We prove that, with high probability,  the  $\ell_{\infty}$ norm of the  estimation error between  the output of GNA  and the underlying target decays to $\mathcal{O}(\sqrt{\frac{\log n }{m}}) $ after  at most $\mathcal{O}(\log s)$ iterations. Moreover, the underlying
support can be recovered with high probability  in  $\mathcal{O}(\log s)$ steps provided that  the target signal is detectable.
Extensive numerical simulations    and comparisons with state-of-the-art methods  are presented  to illustrate the robustness  of our proposed
decoder  and the efficiency  of the GNA  algorithm. 
\end{abstract}

\section{Introduction}\label{sec:intro}
Compressive sensing is a powerful signal acquisition approach  with  which one  can  recover signals beyond bandlimitedness    from noisy under-determined measurements whose number  is closer to the order of the signal complexity
than the Nyquist rate \cite{CandesRombergTao:2006,Donoho:2006,FazelCandesRcht:2008,FoucartRauhut:2013}.
Quantization that transforms the infinite-precision measurements into discrete ones is necessary for storage and transmission \cite{Sayood:2017}.
 Among others, scalar quantization is widely considered  due to its low  computational complexity.
A scalar  quantizer $\mathcal{Q}(\cdot)$ with bit depth $b$ is fully characterized by the
quantization regions  $\{[r_{\ell}, r_{\ell+1})\}_{\ell=1}^{L}$ constituting a partition of $\mathbb{R}$,  where $L = 2^{b} $,  $r_{1}=-\infty, \ \ r_{L+1}=\infty$
as well as the  codebook   $\{\omega_{\ell}\}_{\ell=1}^{L}$, where
$\mathcal{Q}(t)=\omega_{\ell}$  if  $ t \in [r_{\ell}, r_{\ell+1})$.
The 1-bit quantizer $\mathcal{Q}(t)= \mathrm{sign} (t)$,  an extreme case of scalar quantization,
that   codes the measurements  into binary values with  a single bit has been introduced into compressed sensing
\cite{BoufounosBaraniuk:2008}. The 1-bit compressed sensing (1-bit CS)  has drawn much attention because of its low cost in hardware implementation and storage  and its robustness in the low signal-to-noise ratio scenario  \cite{Laska:2012}.
\subsection{Notation and 1-bit CS model}\label{setting}
 We denote by $\Psi_i\in \mathbb{R}^{m\times1}, i = 1,..., n,$ and $\psi_j \in \mathbb{R}^{n\times 1},j=1,...m,$ the $i$th column and $j$th row  of $\Psi$, respectively. We denote zero  vector by   $\textbf{0}$. We use   $[n]$  to denote the set $\{1,...,n\}$, and  $\textbf{I}_n$ to denote the identity matrix of size $n\times n$.
 For  $ A,B\subseteq [n]$ with cardinality  $|A|,|B|$, $x_{A}=(x_{i}, i\in A)\in \mathbb{R}^{|A|}$,  $\Psi_{A}=(\Psi_{i}, i\in A)\in \mathbb{R}^{m\times|A|}$ and $\Psi_{AB}\in \mathbb{R}^{|A|\times
|B|}$  denotes a submatrix of $\Psi$ whose rows and columns
are listed in $A$ and $B$, respectively.
Let  ${x|_A} =(x_i\mathbf{1}_{i\in A})\in \mathbb{R}^n$, where, $\mathbf{1}_{A}$ denotes the indicator function of
set $A$.
Let
$|x|_{s,\infty}$ and $|x|_{\mathrm{\min}}$
be the $s$th largest elements (in absolute value) and the minimum absolute value of $x$, respectively.
We use $\mathcal{N}(\textbf{0},\Sigma)$ to denote the multivariate normal distribution, with $\Sigma$  symmetric and positive definite.
 Let $\gamma_{\textrm{max}}(\Sigma)$ and  $\gamma_{\textrm{min}}(\Sigma)$
 be the largest  and the smallest eigenvalues
of $\Sigma$, respectively.
Let $\mathrm{supp}(x)$ denote the support of $x$.
We use $\|x\|_{\Sigma}$ to denote the elliptic norm of $x$ with respect to $\Sigma$, i.e., $\|x\|_{\Sigma} = (x^{t}\Sigma x)^{\frac{1}{2}}.$ Let $\|x\|_p = (\sum_{i=1}^{n}|x_{i}|^p)^{1/p}, p\in [1,\infty],$  be
the $\ell_p$-norm of  $x$.
We denote the  number of nonzero elements of $x$ by $\|x\|_{0}$. The symbols $\|\Psi\|$ and $\|\Psi\|_{\infty}$ stands for the operator norm of $\Psi$ induced by $\ell_2$ norm and the maximum  pointwise absolute value of $\Psi$, respectively.  $\textrm{sign}(\cdot)$ operates componentwise with $\textrm{sign}(z) =1$ if $z \geq 0$ and $\textrm{sign}(z) =-1$ otherwise,  and $\odot $ denotes  the pointwise Hardmard product.
By  $\mathcal{O}(\cdot)$, we  ignore some positive numerical  constants.

Following \cite{PlanVershynincpam:2013,HuangJiao:2018},  we consider  1-bit CS model
 \begin{equation}\label{setup}
 y = \eta \odot\textrm{sign} (\Psi x^* + \epsilon),
  \end{equation}
where $y\in \mathbb{R}^{m}$ are the binary  measurements, $x^{*}\in \mathbb{R}^{n}$ is an unknown signal with  $\|x^*\|_{0}\leq s$, $\Psi \in \mathbb{R}^{m\times n}$ is a random  matrix whose rows $\psi_i, i \in [m]$  are   i.i.d.  random vectors sampled from   $\mathcal{N}(\textbf{0},\Sigma)$ with  an unknown
covariance matrix $\Sigma $, $\eta\in \mathbb{R}^{m}$ is a random vector modeling the sign flips of $y$ whose  coordinates $\eta_is$ are i.i.d.  satisfying
 $\mathbb{P}[\eta_i = 1] = 1- \mathbb{P}[\eta_i = -1] = p \neq \frac{1}{2},$ and $\epsilon \in \mathbb{R}^{n}$ is a random vector sampled from  $\mathcal{N}(\textbf{0},\sigma^2\textbf{I}_m)$ with an
unknown noise level $\sigma$ modeling errors before quantization.   We assume  $\eta_i, \epsilon_i$ and $\psi_i$ are independent.
 Since $\sigma$ is unknown,  model \eqref{setup} is  unidentifiable  in the sense that  $y = \eta \odot\textrm{sign} (\Psi x^* + \epsilon) =  \eta \odot\textrm{sign} (\alpha\Psi x^* + \alpha \epsilon), \ \ \forall \alpha > 0.$ Therefore,  the best one can do     is to recover $x^*$ up to a positive constant.   Without loss of generality
we assume $\| x^*\|_{\Sigma} = 1$.

 \subsection{Previous work}
 It is a  challenging task  to decode from nonlinear, noisy and even  sign-flipped  binary  measurements.
A lot of  efforts  have been devoted to studying the theoretical and computational issues  in  the 1-bit CS since the  pioneer     work of  \cite{BoufounosBaraniuk:2008}.
 It has been shown that   support and  vector recovery can be guaranteed in both  noiseless and noisy setting provided that
   $m > \mathcal{O}(s\log n)$ \cite{GopiJain:2013,JacquesDegraux:2013,PlanVershynincpam:2013,BauptBaraniuk:2011,JacquesLaska:2013,GuptaNowakRech:2010,BauptBaraniuk:2011,PlanVershynin:2013,ZhangYiJin:2014,Ahsen:2019}, which is the sample complexity  required in the standard CS setting.
  Adaptive sampling are  considered to  improve  the sampling and decoding performance \cite{GuptaNowakRech:2010,DaiShenXuZhang:2016,BaraniukFoucartNeedellPlan:2017}.
   Further refinements  have  been proposed   in the setting of  non-Gaussian measurement  \cite{AiPlanVershynin:2014,Goldstein:2018} and to    recover the magnitude of the target \cite{KnudsonSaabWard:2016,BaraniukFoucartNeedellPlan:2016}.
  Greedy methods \cite{LiuGongXu:2016,Boufounos:2009,JacquesLaska:2013} and first order methods   \cite{BoufounosBaraniuk:2008, LaskaWenYinBaraniuk:2011, YanYangOsher:2012,DaiShenXuZhang:2016}  are developed  to minimize the sparsity promoting nonconvex objected function caused by    the unit sphere constraint or the  nonconvex regularizers.
Convex  relaxation models are also proposed \cite{ZhangYiJin:2014,PlanVershynin:2013,PlanVershynincpam:2013,
ZymnisBoydCandes:2010,PlanVershynin:2017} to address the nonconvex optimization  problem.
Next, we review some of the above mentioned works and make some comparison with our main results.
Assuming  $\Sigma = \mathbf{I}$ and $\sigma = 0$ and $p = 1$, \cite{BoufounosBaraniuk:2008} proposed to decode $x^*$ with \begin{equation*}
\min_{x\in \mathbb{R}^n} \|x\|_1\quad \mathrm{s.t.}\quad y\odot \Psi x \geq  0, \quad \|x\|_2 = 1.
 \end{equation*}
The  Lagrangian version of the above formulation,  i.e.,
   \begin{equation*}
\min_{x\in \mathbb{R}^n}  \|\max\{\textbf{0},-y\odot \Psi x\}\|_2^2 + \lambda\|x\|_1\quad \mathrm{s.t.} \quad \|x\|_2 = 1.
 \end{equation*}
 is  solved  via  first order method \cite{LaskaWenYinBaraniuk:2011}.
  \cite{JacquesLaska:2013} proposed
   \begin{equation}\label{orrobust}
\min_{x\in \mathbb{R}^n}  \mathcal{L}(\max\{\textbf{0},-y\odot \Psi x\}) \quad \mathrm{s.t.} \quad \|x\|_0 \leq s, \quad \|x\|_2 = 1,
 \end{equation}
to deal with    noises, where  $\mathcal{L}(\cdot) = \|\cdot\|_1$ or $\|\cdot\|_2^2$.
  Binary iterative hard thresholding (BITH), a projected sub-gradient method, is developed   to solve  \eqref{orrobust}.
Assuming  $p\neq 1, \sigma = 0$, i.e.,  considering  sign flips in the noiseless model,  \cite{DaiShenXuZhang:2016} proposed
 \begin{equation}\label{prox}
\min_{x\in \mathbb{R}^n}  \lambda \|\max\{\textbf{0},\nu \textbf{1}-y\odot \Psi x\}\|_{0} + \frac{\beta}{2} \|x\|_2^2 \quad \mathrm{s.t.} \quad \|x\|_0 \leq s,
 \end{equation}
 where
 $\nu >  0,  \beta > 0 $ are tuning parameters.
 \cite{YanYangOsher:2012} proposed  adaptive outlier pursuit (AOP) as a generalization of \eqref{orrobust}  to recover $x^*$ and  simultaneously  detect the entries with sign flips via
   \begin{equation*}
\min_{x\in \mathbb{R}^n, \Lambda \in \mathbb{R}^{m}} \mathcal{L}(\max\{\textbf{0},-\Lambda \odot y\odot \Psi x\}) \quad \mathrm{s.t.} \quad \Lambda_i \in \{0,1\},\quad \|\textbf{1} - \Lambda\|_1 \leq N, \quad \|x\|_0 \leq s,  \quad \|x\|_2 = 1,
 \end{equation*}
 where $N$ is the  number of sign flips. Alternating minimization on $x$ and $\Lambda$ are adopted to solve the optimization problem.  \cite{HuangShiYan:2015} considered  both the noises and the sign flips with  pinball loss,
  \begin{equation*}
 \min_{x\in \mathbb{R}^n}  {\mathcal{L}}_{\tau}( \nu \textbf{1}  - y \odot \Psi x\}) \quad \mathrm{s.t.} \quad \|x\|_0 \leq s  \quad \|x\|_2 = 1,
 \end{equation*}
 where $\mathcal{L}_{\tau}(t) = t \textbf{1}_{t\geq 0} - \tau t \textbf{1}_{t<0}$. The pinball iterative hard thresholding   is  developed  to solve the above display.
 In general, there are no theoretical guarantees  for  the  models  mentioned  above except  \cite{JacquesLaska:2013} and  \cite{DaiShenXuZhang:2016}, where they proved
   the estimation error of both \eqref{orrobust} and \eqref{prox} are smaller than $\delta$  provided that   $m \geq \mathcal{O}(\frac{s}{\delta^2} \log n )$.
The aforementioned   state-of-the-art methods   are  nonconvex,   thus it is  hard to justify whether the  corresponding algorithms are loyal to their models. In contrast, in this  paper  we not only  drive the  estimation error of proposed decoder  in Theorem \ref{errsub} which is the same order as those  of \cite{JacquesLaska:2013} and  \cite{DaiShenXuZhang:2016}, but also derive a similar bound on   the  estimation error between $\{x^k\}_{k}$, the output of our generalized Newton algorithm,   and the underlying target $x^*$  in Theorem \ref{th:est}.  Therefore, there is no gap between our  theory and  computation.

Convex relaxation  is  another line of research in 1-bit CS since the seminal    work  \cite{PlanVershynincpam:2013}, where they  proposed  the following linear programming model in the   noiseless setting  without sign flips
 \begin{equation*}
 x_{\mathrm{lp}}\in \arg\min_{x\in \mathbb{R}^n} \|x\|_1 \quad \mathrm{s.t.} \quad  y \odot\Psi x\geq 0  \quad \|\Psi x\|_1  =  m.
 \end{equation*}
 As shown in \cite{PlanVershynincpam:2013},  the estimation error is   $$\|\frac{x_{\mathrm{lp}}}{\|x_{\mathrm{lp}}\|} - x^*\|\leq \mathcal{{O}}(({s \log^2 \frac{n}{s}})^{\frac{1}{5}}).$$
 The above result is improved to $$\|\frac{x_{\mathrm{cv}}}{\|x_{\mathrm{cv}}\|} - x^*\|\leq \mathcal{{O}}(({s \log \frac{n}{s}})^{\frac{1}{4}})$$ in \cite{PlanVershynin:2013}, where both the noises and the sign flips are allowed,  through  considering the convex problem
  \begin{equation}\label{convex1}
 x_{\mathrm{cv}} \in \arg\min_{x\in \mathbb{R}^n} -\langle y, \Psi x\rangle/m \quad \mathrm{s.t.} \quad \|x\|_1 \leq s,  \quad \|x\|_2 \leq 1.
 \end{equation}
  The results derived  in \cite{PlanVershynincpam:2013} and \cite{PlanVershynin:2013} are suboptimal comparing  with   our result in Theorem \ref{errsub}.
 In the noiseless case,  \cite{ZhangYiJin:2014} considered the Lagrangian version of  \eqref{convex1}
  \begin{equation}\label{convex2}
 \min_{x\in \mathbb{R}^n} -\langle y, \Psi x\rangle/m  +  \lambda \|x\|_1 \quad \mathrm{s.t.}  \quad \|x\|_2 \leq 1.
 \end{equation}
 In this special case, the estimation error  derived  matched our  results in  Theorem \ref{errsub}. However, the result derived in \cite{ZhangYiJin:2014} does not hold when     $\Sigma \neq \textbf{I}_n.$
 \cite{PlanVershynin:2017,Vershynin:2015},  proposed a simple projected linear estimator $\mathrm{Proj}_{K}(\Psi^t y/m)$, where $K = \{x\big |\|x\|_1 \leq s, \|x\|_2 \leq 1\}$,  to estimate the low-dimensional structure target belonging to  $K$  in high dimensions from noisy and possibly nonlinear observations. They derived the same  order of estimation error as that in our Theorem \ref{errsub} assuming $\Sigma $ is known. However, as shown in \cite{HuangJiao:2018}, this simple decoder is not roust to noises and sign flips.
 \cite{HuangShiYan:2015} introduced a convex model by  replacing the linear loss in \eqref{convex2} with the  pinball loss and
 \cite{ZymnisBoydCandes:2010} proposed an
 $\ell_1$ regularized maximum likelihood estimate. However, the  sample complexity or estimation error are not studied in  these two  works.
\subsection{Contributions}
In this paper, we study to decode  form the 1-bit CS model \eqref{setup}
 with the cardinality constraint  ordinary  least square
\begin{equation}\label{subreg}
 x_{\ell_0} \in \arg \min  \frac{1}{2m}\|y - \Psi x\|_2^2, \ \ \mathrm{s.t.} \ \ \|x\|_0 \leq s.
  \end{equation}

 (1) We prove  that, with high probability  the estimation error  $\|x_{\ell_0}/c - x^*\| \leq \delta, \delta \in (0,1)$ provided that   $m\geq \mathcal{O}(\frac{s\log n}{\delta^2})$, which is minimax optimal and match  the sample complexity
required for  the standard CS.
\begin{adjustwidth}{1cm}{1cm}
   \textit{Up to a constant $c$, the  sparse signal  $x^*$ can  be decoded from 1-bit measurements with the  cardinality constraint  least squares, as long as the sample complexity is $m \geq \mathcal{O}(s\log n)$.}
 \end{adjustwidth}

 (2) We introduce    a  generalized  Newton algorithm (GNA) to solve the $\ell_0$-constraint  minimization \eqref{subreg}  with computational cost $\mathcal{O}(\max \{s^2, m\}n)$  per iteration.
 We prove that, up to a constant $c$,  with high probability,  the  $\ell_{\infty}$ norm of the  estimation error between $\{x^k\}_{k}$, the output of GNA,  and the target $x^*$ decays to $\mathcal{O}(\sqrt{\frac{\log n}{m}}) $ with at most $\mathcal{O}(\log s)$ iterations.  Moreover, the underlying
support can be recovered with high probability  in  $\mathcal{O}(\log s)$ steps provide that  the target signal is detectable.
{The code is available  at \url{http://faculty.zuel.edu.cn/tjyjxxy/jyl/list.htm}.}

 \begin{adjustwidth}{1cm}{1cm}
   \textit{Up to a constant $c$, with high probability,  the target  $x^*$ can be  decoded   from noisy and sign flipped binary measurements  with a sharp error  via    the proposed generalized  Newton method   costing at most   $\mathcal{O}(n\max \{s^2, m\} \log s)$ floats. Meanwhile, the true support can be recovered with  the cost
   $\mathcal{O}( n\max \{s^2, m\} \log s)$ if $x^*$ is detectable.}
 \end{adjustwidth}

The rest of the paper is organized as follows. In Section 2  we  consider the cardinality constraint least square decoder  and prove a minimax bound on   $\|x_{\ell_0}/c -x^* \|$.
 In Section 3 we introduce the  generalized Newton  algorithm to solve  \eqref{subreg}. Also, we  prove the sharp estimation error of the output of GNA and study its support recovery property.
In Section 4 we conduct numerical simulation   and compare  with  existing  state-of-the-art 1-bit CS  methods.
We conclude  in Section 5. Proofs of the Lemmas,  Theorems and Propositions  are provided  in the Appendix.

\section{Decoding with cardinality constraint  least squares}
 Using least squares to estimate parameters in the scenario of   model  misspecification    goes back to
 \cite{Brillinger:2012}, and see also  \cite{LiDuan:1989} and the references therein  for related  development  in the setting $m \gg n$.
 Recently, with this idea,  \cite{PlanVershynin:2016,Neykov:2016,HuangJiao:2018} proposed Lasso type methods  to estimate parameters from   general under-determined  nonlinear measurements.
 Following  this line, we propose the  cardinality constrained least squares decoder \eqref{subreg}.
  Model \eqref{subreg} and its "Lagrangian" version have been  studied when $y$ is continuous in compressed sensing and high-dimensional statistics   \cite{FoucartRauhut:2013,ZhangZhang:2012,JiaoJinLu:2015}. In the scenario of continuous $y$,   the global minimizer of  \eqref{subreg}  is a unbiased estimator of the target $x^*$,  and has better selection and prediction results than the convex Lasso   model \cite{ZhangZhang:2012,Zhang:2017}.

  As far as we know, this is the first study  of   \eqref{subreg} in the setting of   quantized measurements.
Next, we show that, up to the constant $$c=(2p - 1)\sqrt{\frac{2}{\pi(\sigma^2+1)}},$$ the estimation error of $x_{\ell_{0}}$
achieves a minimax optimal order even if the measurements   are binary and  noisy and  corrupted by sign flips.
 \begin{theorem}\label{errsub}
 Assume
$ n> m \geq  \max\{\frac{4C_1}{C_2^2}\log n, \frac{16(C_2+1)^2}{C_1}s\log\frac{en}{s}\}$,  $s \leq \exp^{(1-\frac{C_1}{2})} n$. Then
with probability at least $1-2/n^3  -4/n^2$, we have,
\begin{equation}\label{lssub}
\|x_{\ell_0}/c - x^*\| \leq  \frac{9(\sigma+1+C_3)}{\sqrt{C_1} \gamma_{\mathrm{min}}(\Sigma)|p-1/2|}\sqrt{\frac{s\log n}{m}}.
\end{equation}
\end{theorem}

 \begin{proof}
The proof  is given in Appendix \ref{app:errsub}.
\end{proof}

\begin{remark}
The sample complexity required in Theorem \ref{errsub}, i.e.,
$m \geq \mathcal{O} (s\log\frac{n}{s})$ is  optimal to guarantee the possibility of decoding  from   binary measurement successfully \cite{Ahsen:2019}. The estimation error  derived in in Theorem \ref{errsub}  matches the minimax optimal order  $\mathcal{O}(\sqrt{\frac{s \log n}{m}})$ in the sense that  it is the optimal order that can be attained even if  the signal  is measured precisely without  quantization \cite{Raskutti:2011}.
 Theorem \ref{errsub} also implies that  the support of $x_{\ell_{0}}$ coincides with that of $x^*$ as long as   the minimum nonzero  magnitude of $x^*$ is large enough, i.e., $|x^*|_{\min} \geq \mathcal{O}(\sqrt{\frac{s\log n}{m}})$.

 Comparing with Theorem 3.1 in \cite{HuangJiao:2018}, the estimation error of the cardinality constraint least squares decoder   proved here is better than that of Lasso  in the sense that  it does not depend on the condition number of  $\Sigma$. Meanwhile  the number of samples needed here is smaller than that in \cite{HuangJiao:2018}.  Both improvements  are verified by  our  numerical studies, see Section 4.
\end{remark}

\section{Generalized Newton  algorithm}
 In this section we develop a generalized Newton algorithm  (GNA) to  solve \eqref{subreg} approximately.
 Furthermore, we bound  the  $\ell_{\infty}$ norm of the  estimation error between  the output of GNA  and the target $x^*$ and study its  support  recovery  property.
\subsection{KKT condition and derivation of GNA}
We first derive a KKT condition of \eqref{subreg}, which is our starting point for deriving GNA.
We use $x$ to denote $x_{\ell_0}$ for simplicity.
\begin{lemma}\label{kkt}
Let $\eta \in (0,\frac{4}{9\gamma_{\mathrm{max}}(\Sigma)})$,
under the condition of Theorem \ref{errsub}, we have with probability at least $1-4/n^2$,
\begin{equation}\label{kkteq}
\left\{
                             \begin{array}{ll}
                               d = \Psi^t(y-\Psi x)/m, \\
                               x = \mathcal{H}_{s}(x + \eta d),
                             \end{array}
                           \right.
\end{equation}
where,  $\mathcal{H}_{s}(z)$ is the hard thresholding operation on $z$ that keeps the first $s$ largest entries  in absolute value and kill others as zero.
\end{lemma}

\begin{proof}
The proof is shown in in Appendix \ref{app:kkt}.
\end{proof}

Let $A =\text{supp}(x)$ and $I=\bar{A}$.
By  \eqref{kkteq} and the definition of $\mathcal{H}_{s}(\cdot)$, we have
\begin{equation*}
A=\{i\in[n]:|x_{i}+ \eta d_{i}|\geq |x+d|_{s,\infty}\}, \ \ I=\bar{A},
\end{equation*}
and
\begin{align}
\label{eq5a}
\left\{
\begin{aligned}
&x_{I} =\mathbf{0}\\
&d_{A} =\mathbf{0}\\
&x_{A}=(\Psi_{A}^{t}\Psi_{A})^{-1}\Psi_{A}^{t}y\\
&d_{I} =\Psi_{I}^{t}(y-\Psi_{A}x_{A})/m.
\end{aligned}
\right.
\end{align}
Let $(x^{k},d^{k})$ be the values at the $k$-th iteration,
and let $\{A^{k},I^{k}\}$ be the active and inactive sets defined as
\begin{equation}\label{eq4}
A^k=\{i\in[n]:|x^k_{i}+ \eta d^k_{i}|\geq|x^k+\eta d^k|_{s,\infty}\}, \ \
I^k= \overline{A^k}.
\end{equation}
 Our proposed generalized Newton  algorithm updates the primal and dual pair $(x^{k+1},d^{k+1})$
according to  \eqref{eq5a} as follows:
\begin{align}\label{eq5}
\left\{
\begin{aligned}
&x^{k+1}_{I^{k}}=\mathbf{0}\\
&d^{k+1}_{A^{k}} =\mathbf{0}\\
&x^{k+1}_{A^{k}}=(\Psi_{A^k}^{t}\Psi_{A^k})^{-1}\Psi_{A^k}^{t}y\\
&d^{k+1}_{I^{k}}=\Psi_{I^k}^{t}(y -\Psi_{A^k}x_{A^k}^{k+1})/m\\
\end{aligned}
\right.
\end{align}

We summarize the GNA in detail in the following  Algorithm.
{
\begin{algorithm}
   \caption{Generalized Newton Algorithm (GNA)}\label{alg:genew}
   \begin{algorithmic}[1]
     \STATE Input $y,\Psi,s, \eta$,  initial guess $x^0$, maximum number of iteration \textsl{MaxIter}. Let $d^0= \Psi^t( y -\Psi x^0)/m$.
     \FOR {$k=0,1,...\textsl{MaxIter}$}
     \STATE Compute the active and inactive sets $A^k$ and $I^k$ respectively by \eqref{eq4}.
     \STATE Update $x^{k+1}$ and $d^{k+1}$  by
       \eqref{eq5}.
     \STATE If ${A}^{k} = {A}^{k+1}$, stop.
     \ENDFOR
   \STATE Output $x^{k+1}$.
   \end{algorithmic}
\end{algorithm}
}

\begin{remark}\label{ccp}
 It takes $\mathcal{O}(n)$
flops to finish step  3  in GNA. In step 4, it takes   $\mathcal{O}(mn)$ flops except the least squares step, which is  the most time consuming part. Forming the matrix $\Psi_{{A}^{k}}^{t}
\Psi_{{A}^{k}}$  takes $\mathcal{O}(ms^2)$ flops while the cost of computing  $\Psi^t y$ is negligible since it
can be precomputed and stored. Inverting $\Psi_{{A}^{k}}^{t}
\Psi_{{A}^{k}}$ costs  $ \mathcal{O}(s^3)$ flops by direct methods.
Therefore,  the overall cost of GNA  per iteration is  $\mathcal{O}(\max\{n,s^2\}m)$.
\end{remark}

\subsection{GNA as Newton type method}
When $\eta = 1$,  GNA Algorithm \ref{alg:genew}  has been proposed as a greedy method for standard compressed sensing \cite{Foucart:2011} with the name hard threshoding pursuit.
Interestingly,
we show that the proposed GNA algorithm \ref{alg:genew} can be interpreted as Newton type method for finding roots of the KKT system \eqref{kkteq} even though the original problem   \eqref{subreg} is nonconvex and nonsmooth.
Let $ w =(x; d)$ and
$
F(w)=
\left(
\begin{array}{c}
 F_{1}(w)\\
F_{2}(w)
\end{array}
 \right)
:\mathbb{R}^{n} \times \mathbb{R}^{n} \rightarrow \mathbb{R}^{2n},
$
where
$
  F_{1}(w)= x  - \mathcal{H}_{s}(x + d)$ and
 $ F_{2}(w)= \Psi^t \Psi x + m d -\Psi^t y.$
\begin{proposition}\label{eqn}
The  iteration in \eqref{eq5} can be equivalently  reformulated as
$$ w^{k+1} = w^k - (H^k)^{-1}F(w^k),$$
where
\begin{equation*}
H^k = \left(
        \begin{array}{cc}
            H^k_{1} & H^k_{2}\\
            \Psi^t\Psi & m\textbf{I} \\
          \end{array}
        \right)
, \quad
H^k_{1} = \left(
          \begin{array}{cc}
            \textbf{0}_{{A}^k{A}^k} & \textbf{0}_{{A}^k{I}^k} \\
            \textbf{0}_{{I}^k{A}^k}& \textbf{I}_{{I}^k{I}^k}\\
          \end{array}
        \right)
 \quad \textrm{and} \quad
  H^k_{2} = \left(
          \begin{array}{cc}
            -\textbf{I}_{{A}^k{A}^k} & \textbf{0}_{{A}^k{I}^k}\\
            \textbf{0}_{{I}^k{A}^k}&  \mathbf{0}_{{I}^k{I}^k}\\
          \end{array}
        \right).
 \end{equation*}
\end{proposition}

\begin{proof}
The proof is provided  in in Appendix \ref{app:eqn}.
\end{proof}

\subsection{Estimation error and support recovery of GNA}
As expected, GNA may exhibit fast local convergence to the KKT point of \eqref{kkteq} since it is  a Newton type method as shown above.
However, in this subsection we consider  in the perspective of  studying  the  estimation error  of GNA, i.e, bounding the
error between $x^k$ and the target signal $x^*$ directly.
Define
\begin{equation}\label{CC}
 C_{*}= \inf_{\|v\|_0 \leq 2s }\frac{ v^{t}\Psi^t\Psi v}{n\|v\|_1\|v\|_{\infty}}, \ \ C^* = \sup_{\|v\|_0 \leq 2s }\frac{ v^{t}\Psi^t\Psi v}{n\|v\|_1\|v\|_{\infty}}.
\end{equation}
\begin{theorem}\label{th:est}
Assume $ n> m \geq  \max\{\frac{4C_1}{C_2^2}\log n, \frac{16(C_2+1)^2}{C_1}s\log\frac{en}{s}\}$,  $s \leq \exp^{(1-\frac{C_1}{2})} n$.
Let $\eta \in (0, \frac{4}{9\gamma_{\mathrm{max}}(\Sigma)\sqrt{s}})$ and  $x^0=\mathbf{0}$ in  GNA.
Then, with probability at least $1-2/n^3  -6/n^2$,
\begin{align}\label{betaerror}
\|x^k/c-x^*\|_{\infty}\leq
\frac{9(1+\sigma +C_3)}{\sqrt{C_1}C_*|p-1/2|}\sqrt{\frac{\log n}{m}},
\end{align}
as long as $k\geq \log_{1/\zeta}(\frac{sm}{\log n} \frac{C^*C_1|c|^2}{16C_*(1+|c|C_3)^2\gamma_{\mathrm{min}}(\Sigma)})$,
where  $\zeta =1-\frac{2\eta C_*(1-\eta \sqrt{s}C^*)}{ \sqrt{s}(1+s)}\in(0,1).$
\end{theorem}

\begin{proof}
The proof is provided  in in Appendix \ref{app:est}.
\end{proof}

\begin{remark}
When $\eta = 1$, our proposed GNA has been studied in the setting of standard compressed sensing \cite{Foucart:2011} and high-dimensional statistics \cite{HuangJiaoLu:2018} by assuming $\Psi$ satisfying   restricted isometry property  condition and sparse Riesz condition, respectively.  In comparison,  the result derived in Theorem  \ref{th:est} does not require such stronger assumptions.  The only requirement to guarantee Theorem \ref{th:est}  is choosing the step size $\eta$ such that $\zeta \in (0,1)$.
Indeed, observing $C^*$ and $C_*$ can be bounded from above by $C_{2s,\mathrm{max}}$ and from below by $C_{2s,\mathrm{min}}/\sqrt{s}$,  the requirement $\eta \in (0,1)$   always holds as long as $C_{2s,\mathrm{min}}>0$ and $C_{2s,\mathrm{max}}< + \infty.$

The number of iterations of GNA  is $\mathcal{O} (\log s)$.
Combing the computational complexity per iteration in Remark \ref{ccp}, we deduce
that  up to a constant $c$, the output  of GNA  with achieve a sharp estimation error of the order $\mathcal{O} (\sqrt{\frac{\log n}{m}})$ with total cost
$\mathcal{O}(\max\{n,s^2\}m \log s).$
\end{remark}

As a consequence of Theorem \ref{th:est}, we can deduce that  the stopping criterion  of GNA will hold in $\mathcal{O}(\log (s))$ steps as long as
the magnitude of the minimum value of the target signal is detectable, i.e.,  $|x^*|_{\mathrm{min}} \geq \mathcal{O}(\sqrt{\frac{\log n}{m}})$. Meanwhile, when GNA stops   the recovered support coincides with $\mathrm{supp}(x^*)$.
\begin{proposition}\label{recsupp}
Suppose
$|x^*|_{\mathrm{min}} > \frac{9(1+\sigma +C_3)}{\sqrt{C_1}C_*|p-1/2|}\sqrt{\frac{\log n}{m}}$ and $|\mathrm{supp}(x^*)| =s$. Under the assumption of Theorem \ref{th:est}, we have with probability at least
$1-2/n^3  -6/n^2$,
$A^k=A^{k+1} = \mathrm{supp}(x^*)$ provided that $k\geq \log_{1/\zeta}(\frac{sm}{\log n} \frac{C^*C_1|c|^2}{16C_*(1+|c|C_3)^2\gamma_{\mathrm{min}}(\Sigma)})$.
\end{proposition}

\begin{proof}
The proof is provided  in in Appendix \ref{app:sup}.
\end{proof}

\begin{remark}
The support recovery property of hard thresholding pursuit
has been studied in the setting of sparse regression \cite{YuanZhang:2017, ShenLi:2017}
 under the assumption  the minimum magnitude of the target is larger than $\mathcal{O}(\sqrt{\frac{s\log(n)}{m}})$,  which is stronger than the requirement in Proposition \ref{recsupp}.
 Meanwhile, the iteration complexity for support recovery in  \cite{ShenLi:2017} is  $\mathcal{O}(s)$, which is suboptimal than the  $\mathcal{O} (\log s)$ complexity in  Proposition \ref{recsupp}.
\end{remark}

\section{Numerical simulations}\label{num}

In this section we show the performance of our proposed cardinality constraint least square   \eqref{subreg} and the GNA Algorithm \ref{alg:genew}.
All the computations were performed on a four-core laptop with 2.90 GHz and 8 GB RAM using \texttt{MATLAB} 2018a.
The \texttt{MATLAB} package \texttt{1-bitGNA} for reproducing all the numerical results is available  at
\url{http://faculty.zuel.edu.cn/tjyjxxy/jyl/list.htm}.

\subsection{Experiment setup}
First we describe the  data generation setting  and  the hyperparameter  choice. In all numerical simulation  the true   signal $x^*$ with
$\|x^*\|_0 = s$ and $\|x^*\|_2 = 1$ is given, and the binary measurements   $y$ are  generated by
 $y = \eta \odot\textrm{sign} (\Psi x^* + \epsilon)$,
where the rows of $\Psi$ are i.i.d.  samples drawn  from  $\mathcal{N}(\textbf{0},\Sigma)$ with $\Sigma_{jk} = \nu^{|j-k|}, 1 \le j, k \le n$  ($0^{0} = 1$).  The $\epsilon$ is sampled from  $\mathcal{N}(\textbf{0},\sigma^2\textbf{I}_m)$,
$\eta\in \mathbb{R}^{m}$ has independent coordinate $\eta_i$
with $\mathbb{P}[\eta_i = 1] = 1- \mathbb{P}[\eta_i = -1] = p$.  We use $(m,n,s,\nu,\sigma,p)$ to denote the data generated as   above description.
We set the initial value  $x^0 = \mathbf{0}$, the step size $\eta = 0.9$ and $\textsl{MaxIter} = 5$  in GNA unless indicated otherwise.   All the  simulation results   are based on 100 independent replications except the last  example.

\subsection{Number of iteration of GNA}
In this subsection, we set $\textsl{MaxIter} = 10$  in GNA.
Figure \ref{fig:iter} shows the average number of iterations of GNA   on data sets on data set  $(m = 500,n = 1000,s = 1:2:20,\nu = 0.1,\sigma= 0.05, p = 1\%).$
We see that the average number of iterations are less than  $4$  as the sparsity level $s$ varying  from $1$ to $20$, which verifies   the $\mathcal{O}(\log s)$ iteration complexity derived in Theorem \ref{th:est} and Proposition \ref{recsupp}.
\begin{figure}[ht!]
  \centering
   \includegraphics[trim = 0cm 0cm 0cm 0cm, clip=true,width=.35\textwidth]{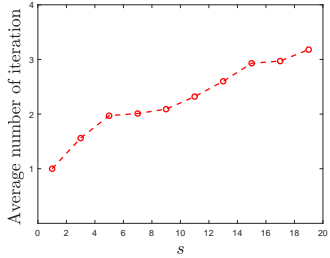}
   \caption{Average umber of iterations  v.s. $s$  on data set  $(m = 500,n = 1000,s = 1:2:20,\nu = 0.1,\sigma= 0.05, p = 1\%)$.} \label{fig:iter}
\end{figure}

\begin{figure}[ht!]
  \centering
  \begin{tabular}{cc}
   \includegraphics[trim = 0cm 0cm 0cm 0cm, clip=true,width=.35\textwidth]{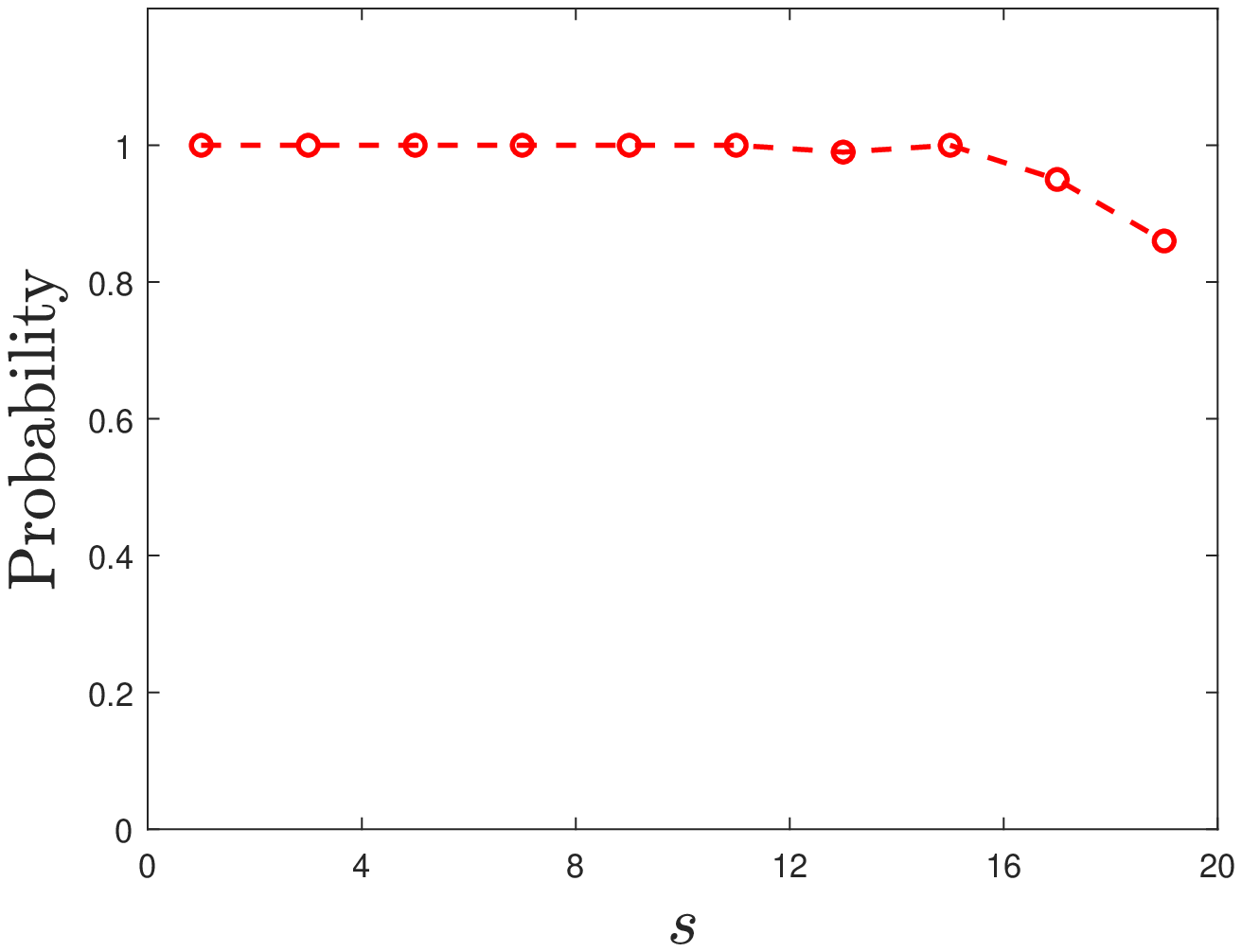} &
   \includegraphics[trim = 0cm 0cm 0cm 0cm, clip=true,width=.35\textwidth]{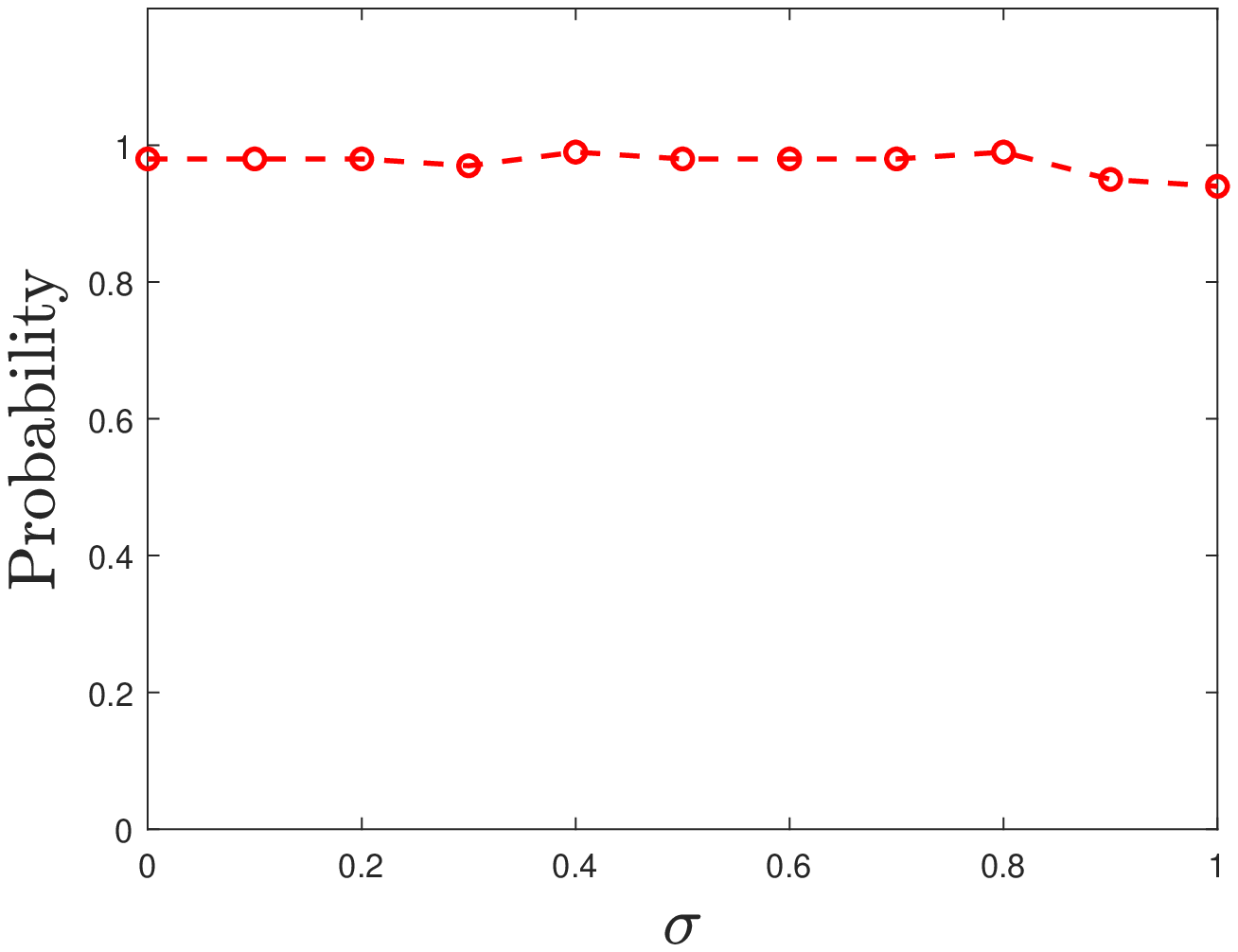} \\
   (a) $s$ & (b) $\sigma$ \\
   \includegraphics[trim = 0cm 0cm 0cm 0cm, clip=true,width=.35\textwidth]{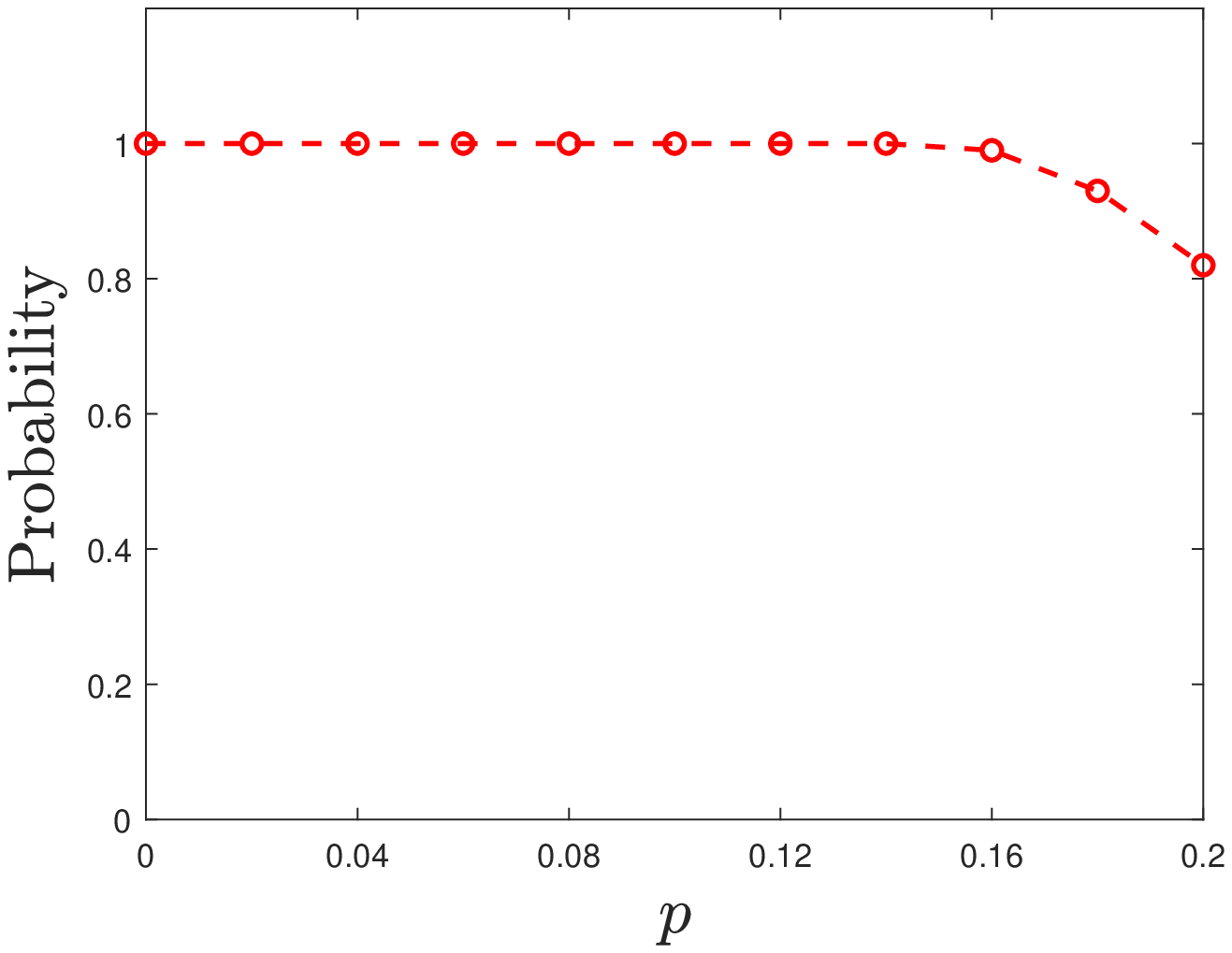} &
   \includegraphics[trim = 0cm 0cm 0cm 0cm, clip=true,width=.35\textwidth]{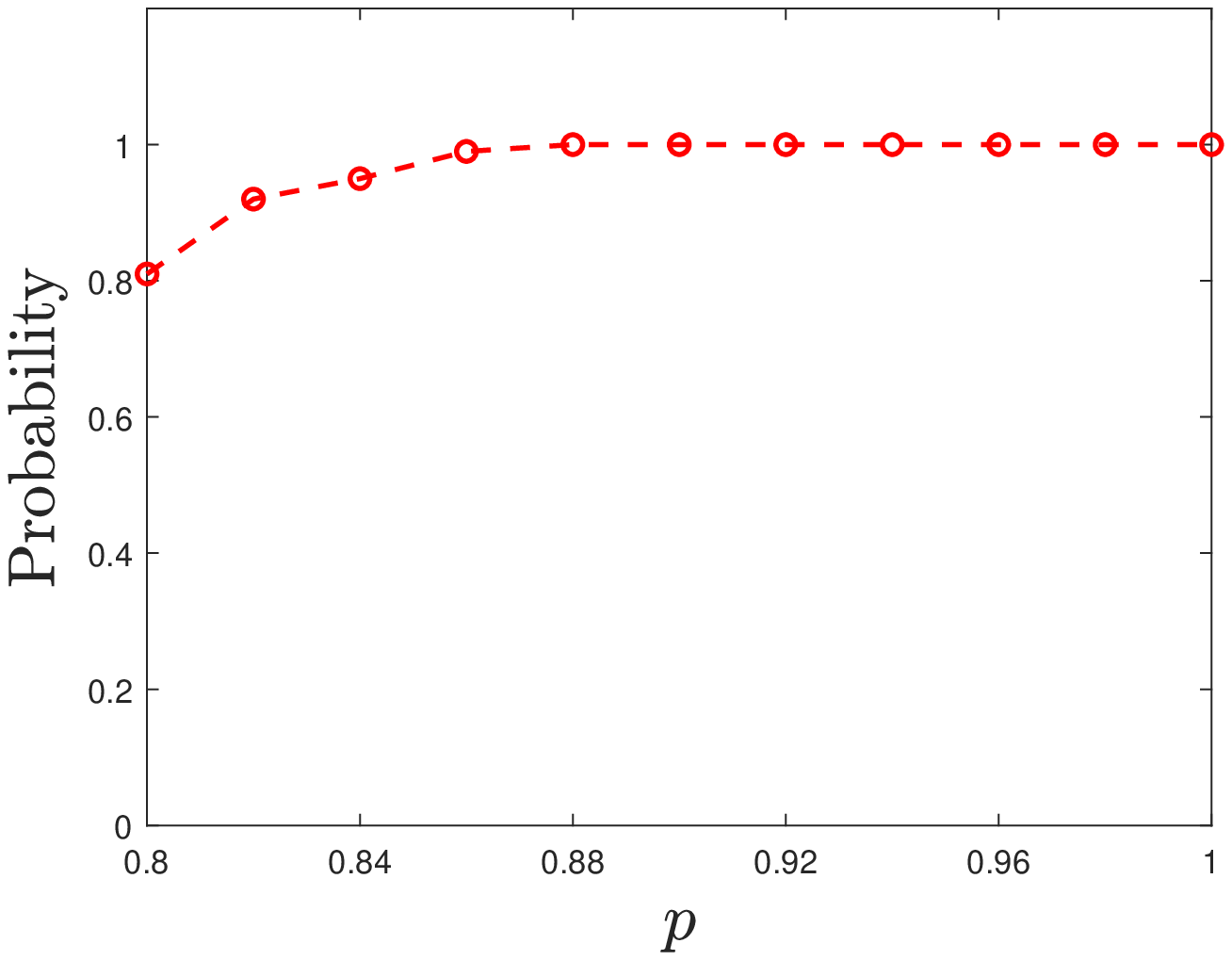}\\
     \\ (c) $p$ & (d) $p$
  \end{tabular}
   \caption{The exact support recovery probability v.s. $s$, $\sigma$ and $p$ on data set  $(m = 500,n = 1000,s = 1:2:20,\nu = 0.1,\sigma= 0.05, p = 1\%)$ (panel (a)), $(m = 500,n = 1000,s = 10,\nu = 0.3,\sigma= 0:0.1:1, p = 5\%)$ (panel (b)) and $(m = 500,n = 1000,s = 5,\nu = 0.1,\sigma= 0.05,p = 0:2\%:20\%)$ (panel (c)), $(m = 500,n = 1000,s = 5,\nu = 0.1,\sigma= 0.05,p = 80\%:2\%:100\%)$  (panel (d)). }\label{fig:probsupp}
\end{figure}

\subsection{Support recovery}
In this subsection we verify the support recovery property of GNA by studying   how the exact support recovery probability depends on   the sparsity level $s$, the noise level $\sigma$ and the of sign flips probability $p$.
We test on date sets $(m = 500,n = 1000,s = 1:2:20,\nu = 0.1,\sigma= 0.05, p = 1\%)$, $(m = 500,n = 1000,s = 10,\nu = 0.3,\sigma= 0:0.1:1, p = 5\%)$,  $(m = 500,n = 1000,s = 5,\nu = 0.1,\sigma= 0.05, p = 0:2\%:20\%)$, $(m = 500,n = 1000,s = 5,\nu = 0.1,\sigma= 0.05, p = 80\%:2\%:100\%)$ and show the corresponding
results  in panel (a)-(d) of Figure \ref{fig:probsupp}  respectively.
 As indicates  by Figure \ref{fig:probsupp},  GNA recovers the  underlying true support
with high probability  as long the sparsity level $s \leq \mathcal{O}(m/{\log n})$   even if the binary  measurements is noisy  and sign flipped. This confirms the theoretical investigations in Proposition  \ref{recsupp}.

\subsection{Comparison with state-of-the-art}
Now we compare our proposed model \eqref{subreg} and  GNA Algorithm \ref{alg:genew} with several state-of-the-art methods such as BIHT  \cite{JacquesDegraux:2013} (\url{http://perso.uclouvain.be/laurent.jacques/index.php/Main/BIHTDemo}), AOP \cite{YanYangOsher:2012} and PBAOP \cite{HuangShiYan:2015} (both AOP and PBAOP   require to know  the sign flips probability $p$,  available  at  \url{http://www.esat.kuleuven.be/stadius/ADB/huang/downloads/1bitCSLab.zip}) and  linear projection (LP) \cite{Vershynin:2015,PlanVershynin:2017}, PDASC \cite{HuangJiao:2018} (\url{http://faculty.zuel.edu.cn/tjyjxxy/jyl/list.htm/}).
We use data set $(m = 500,n = 2500,s = 5,\nu = 0.2,\sigma= 0.2, p = 5\%)$, $(m = 500,n = 2500,s = 5,\nu = 0.3,\sigma= 0.3, p = 10\%)$, $(m = 500,n = 2500,s = 5,\nu = 0.5,\sigma= 0.5, p = 15\%)$,          and $ (m = 1000,n = 5000,s = 10,\nu = 0.2,\sigma= 0.2, p= 5\%)$, $ (m = 1000,n = 5000,s = 10,\nu = 0.3,\sigma= 0.3, p = 10\%)$, $ (m = 1000,n = 5000,s = 10,\nu = 0.5,\sigma= 0.5, p = 15\%)$.
 The average CPU time in seconds (Time (s)), the  average of the  $\ell_2$  error $\|\frac{\hat{x}}{\|\hat{x}\|} - \frac{x^*}{\|x^*\|}\|$ ($\ell_2$-Err) where $\hat{x}$ denotes the output of all the above mentioned methods, and the  probability of  exactly  recovering  true support (PrE ($\%$)) are reported in Table \ref{tab:compother}.
 As shown in  Table \ref{tab:compother}, our GNA
 outperforms all the other state-of-the-art methods in terms of accuracy  ($\ell_2$-Err   and  PrE $(\%)$)  in all the settings.  Meanwhile, GNA  and LP
 archive the best performance on speed.
\begin{table}[ht!]
  \caption{Comparison GNA with state-of-the-art methods on  CPU time in seconds (Time (s)), average $\ell_2$  error $\|\frac{\hat{x}}{\|\hat{x}\|} - \frac{x^*}{\|x^*\|}\|$ ($\ell_2$-Err), probability on exactly  recovering of  true support (PrE ($\%$)).}
 \label{tab:compother}
  \vspace{-0.3cm}
  \begin{center}
  \scalebox{.85}{
  \begin{tabular}{cccccccccccc}
   \hline \hline
  \multicolumn{11}{c}{\quad \quad\quad$(m = 500,n = 2500,s = 5)$}\\
 \hline
  \multicolumn{4}{c}{\quad \quad \quad \quad (a) $(\nu = 0.2,\sigma= 0.2, p = 5\%)$}
  &&\multicolumn{3}{c}{(b) $(\nu = 0.3,\sigma= 0.3, p = 10\%)$}&&\multicolumn{3}{c}{(c) $(\nu = 0.5,\sigma= 0.5, p = 15\%)$}\\
  \cline{2-4}  \cline{6-8} \cline{10-12}
  Method           &Time (s)         & $\ell_2$-Err   & PrE $(\%)$&   & Time (s)  & $\ell_2$-Err    & PrE $(\%)$   &  & Time       &$\ell_2$-Err    &PrE               \\
   BIHT          &       3.39e-1     &    4.98e-1         & 33  &   &     3.72e-1 &        7.97e-1 &     1 &      &          3.39e-1 &         1.10e-0   &      0 \\
   AOP           &       9.18e-1     &    1.59e-1         & 98  &  &      9.86e-1 &        2.05e-1 &     92  &     &         9.28e-1 &         5.30e-1   &     30\\
   LP            &       2.21e-2     &    4.22e-1         & 96 &   &      2.40e-2 &        4.36e-1 &     81&     &           2.17e-2 &         5.43e-1   &     6 \\
   PBAOP         &       3.40e-1     &    1.56e-1         & 98 &   &      3.65e-1 &        2.06e-1 &     93 &     &          3.44e-1 &         5.56e-1   &     25  \\
   PDASC         &       9.49e-2     &    9.56e-2         & 98 &    &     9.10e-2 &        1.73e-1 &     86 &     &          7.05e-2 &         5.25e-1   &     24 \\
    GNA &  \textbf{1.53e-2}  &  \textbf{8.82e-2}   & \textbf{100} &   &     \textbf{1.73e-2} &  \textbf{1.15e-1} &    \textbf{99}  & & \textbf{1.48e-2 }& \textbf{2.15e-1} & \textbf{82} \\
  \hline
  \hline
  \multicolumn{11}{c}{\quad \quad\quad$(m = 1000,n = 5000,s = 10)$}\\
 \hline
  \multicolumn{4}{c}{\quad \quad \quad \quad (a) $(\nu = 0.2,\sigma= 0.2, p = 5\%)$}
  &&\multicolumn{3}{c}{(b) $(\nu = 0.3,\sigma= 0.3, p = 10\%)$}&&\multicolumn{3}{c}{(c) $(\nu = 0.5,\sigma= 0.5, p = 15\%)$}\\
  \cline{2-4}  \cline{6-8} \cline{10-12}
  Method      &Time (s)         & $\ell_2$-Err   & PrE $(\%)$&   & Time (s)  & $\ell_2$-Err    & PrE $(\%)$   &  & Time      &$\ell_2$-Err    &PrE       \\
  BIHT   & 1.35e-0           &  5.33e-1   &  10 &   &             1.32e-0 &             8.06e-1 &     0 &     &     1.29e-0 &          1.03e-0&    0\\
   AOP   & 3.77e-0           &  1.68e-1    &  98  &   &           3.69e-0&              2.13e-1 &     88 &   &      3.52e-0 &          5.30e-1&    7\\
   LP    & \textbf{5.64e-2}  &  4.41e-1    &  82 &   &            \textbf{5.83e-2} &    4.46e-1 &     62 &    & \textbf{ 5.23e-2} &    5.88e-1&    1 \\
   PBAOP & 1.37e-0           &  1.59e-1    &  \textbf{100} &   &  1.32e-0 &             2.12e-1 &     89 &   &      1.29e-0 &          5.81e-1&    5\\
   PDASC &3.64e-1            &  9.77e-2    &  98 &   &            3.05e-1 &             1.84e-1 &     79 &    &     2.52e-1 &          7.72e-1&    4\\
     GNA &6.02e-2            &  \textbf{9.62}e-2 &  \textbf{100}&&5.93e-2 & \textbf{1.24e-1} &   \textbf{99}  &&    5.94e-2 &          \textbf{2.66e-1}&   \textbf{59}\\
  \hline
  \hline
  \end{tabular}}
  \end{center}
\end{table}

Last, we compare our  GNA  with  the aforementioned competitors   to
recover a one-dimensional  signal and two-dimensional image from quantized measurements.
The true  signal and image   are sparse  under
 wavelet basis ``Db1" \cite{Mallat:2008}. Thus,  the matrix $\Psi$  consist
of   random Gaussian matrix  and an  Harr wavelet transform with level  $1$ and  $2^{12}$, respectively.
The target coefficients have  $36$ and $1138$  nonzeros,
and the size $\Psi$ are   $2500\times 8000$ and $5000\times 128^2$.   We set  $\sigma= 0.5$, $p=6\%$ and $\sigma= 0.05$, $p=1\%$ and
the measurements are quantized with bit depth $1$ and $6$ for the signal and image respectively.     The recovered results are shown in
Figure  \ref{fig:1d}, Table \ref{tab:1d} and Figure  \ref{fig:2d}, Table \ref{tab:2d} respectively.  The decoding  results  by   GNA
are   visually more appealing than others, as shown in  Figure  \ref{fig:1d}-\ref{fig:2d}, which are  further confirmed by the PSNR value reported in Table \ref{tab:1d}-\ref{tab:2d},    defined as
$\mathrm{PSNR}=10\cdot \log\frac{V^2}{\rm MSE}$,
where $V$ is the maximum absolute value of the true signal/image, and MSE is the mean
squared error of the reconstruction.

\begin{figure}[ht!]
  \centering
  \begin{tabular}{ccc}
    \includegraphics[trim = 1cm 0.5cm 0.5cm 0cm, clip=true,width=3.8cm]{{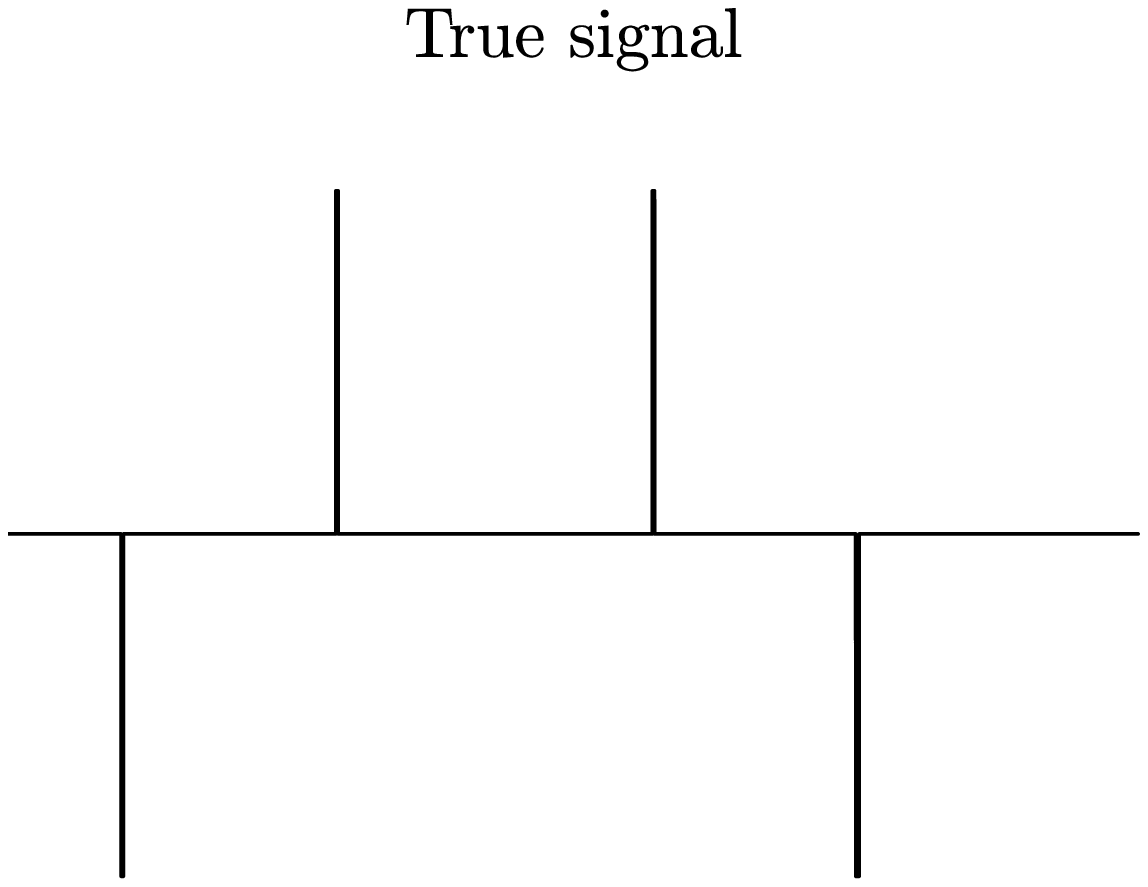}} &
    \includegraphics[trim = 1cm 0.5cm 0.5cm 0cm, clip=true,width=3.8cm]{{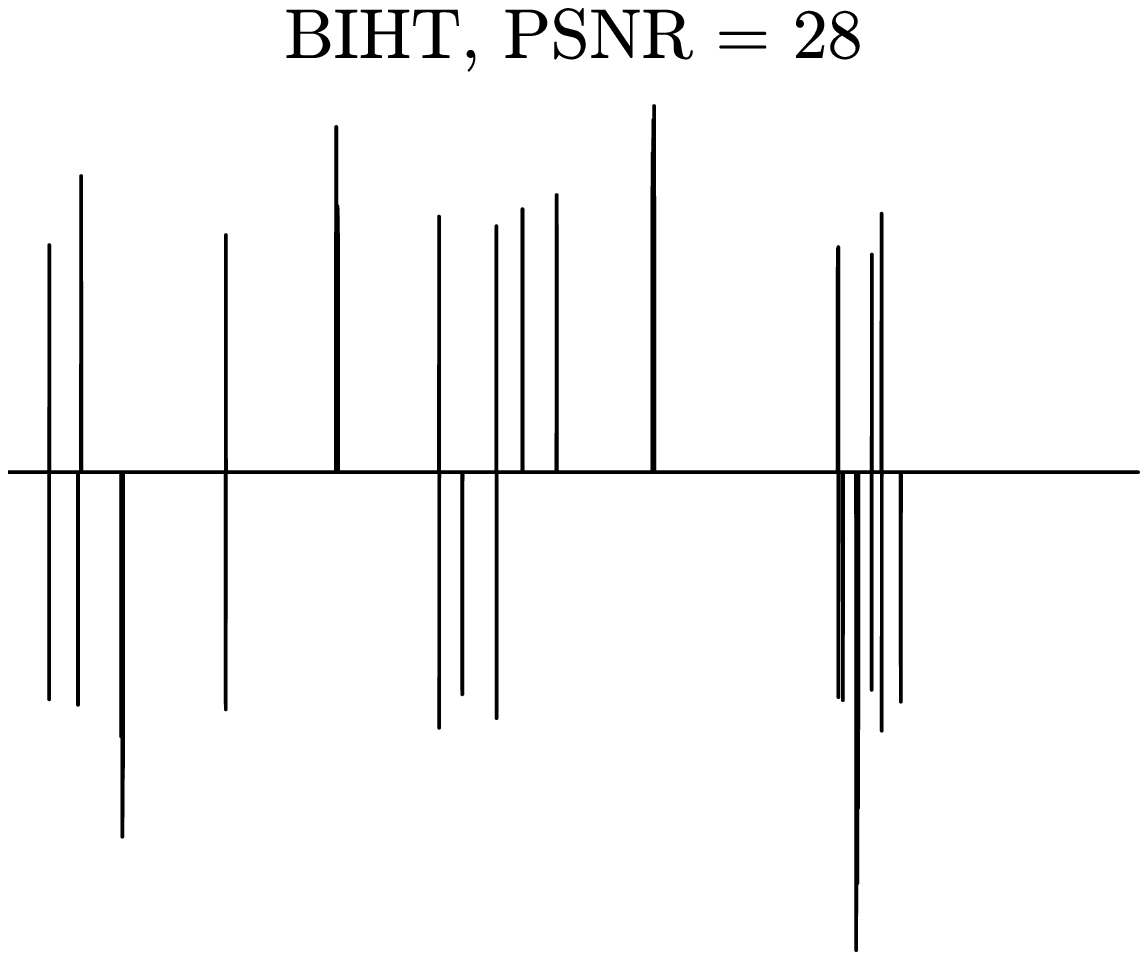}}&
    \includegraphics[trim = 1cm 0.5cm 0.5cm 0cm, clip=true,width=3.8cm]{{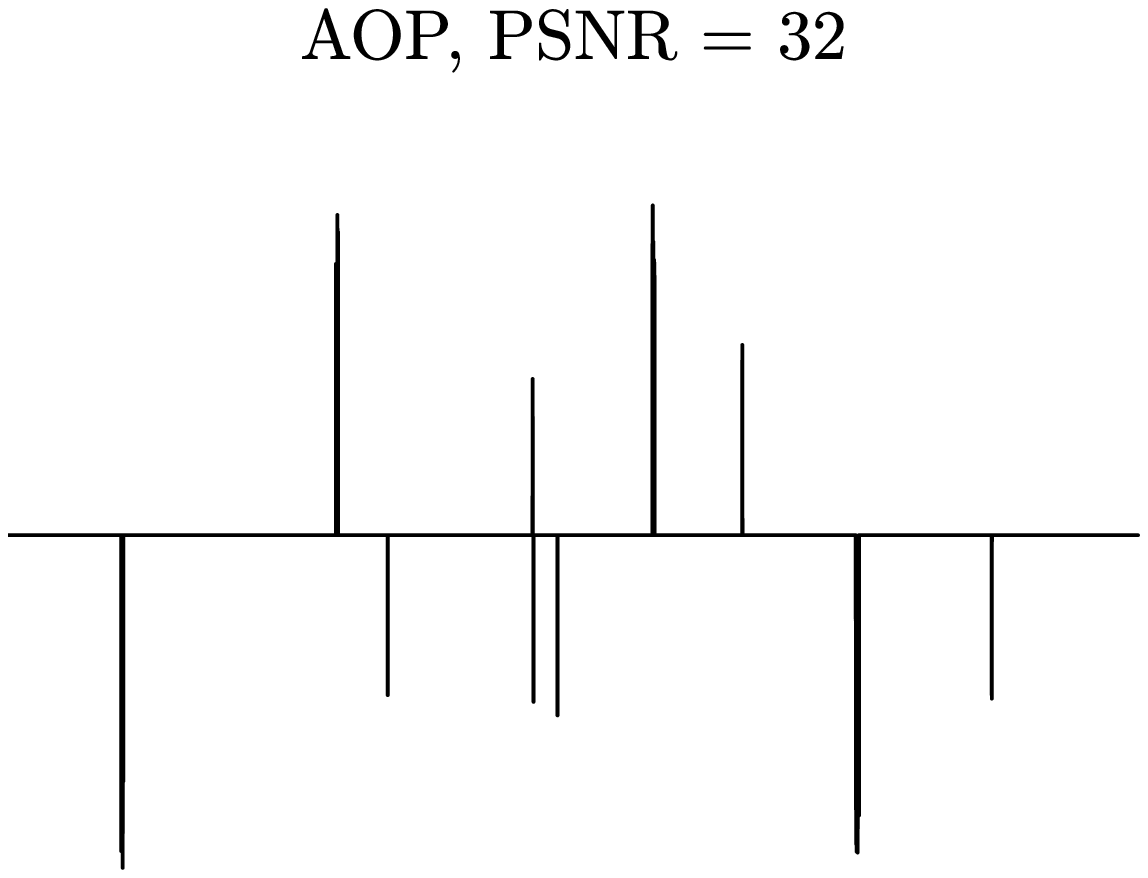}} \\
    \includegraphics[trim = 1cm 0.5cm 0.5cm 0cm, clip=true,width=3.8cm]{{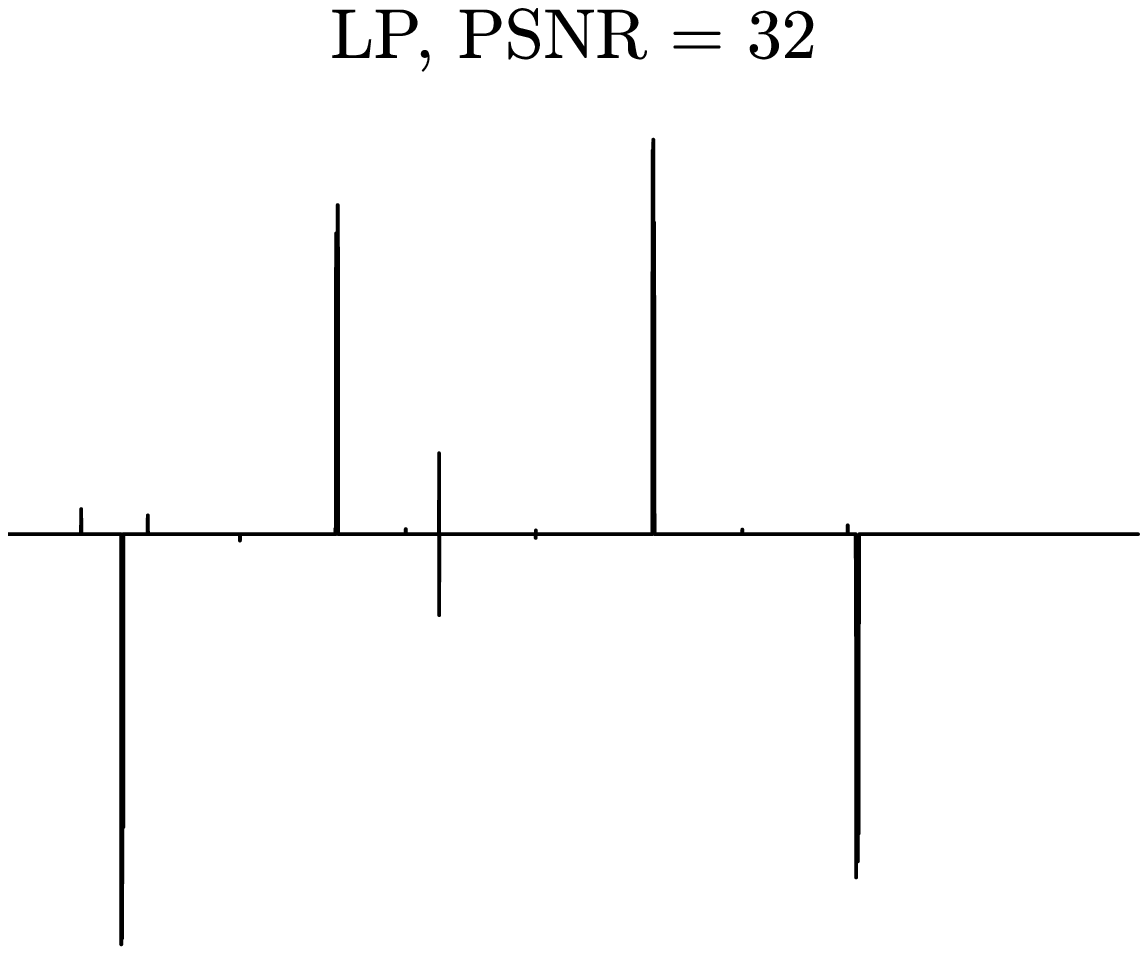}}&
    \includegraphics[trim = 1cm 0.5cm 0.5cm 0cm, clip=true,width=3.8cm]{{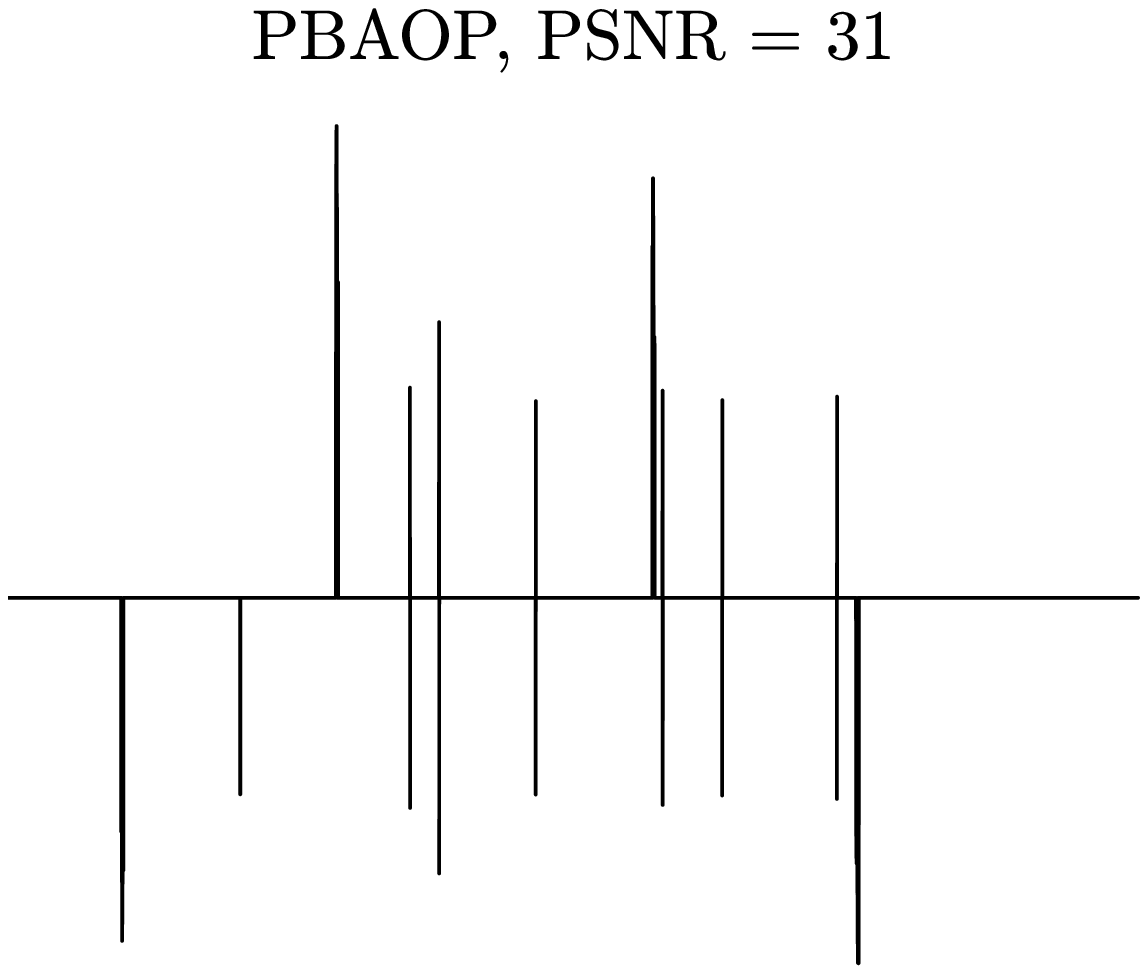}} &
    \includegraphics[trim = 1cm 0.5cm 0.5cm 0cm, clip=true,width=3.8cm]{{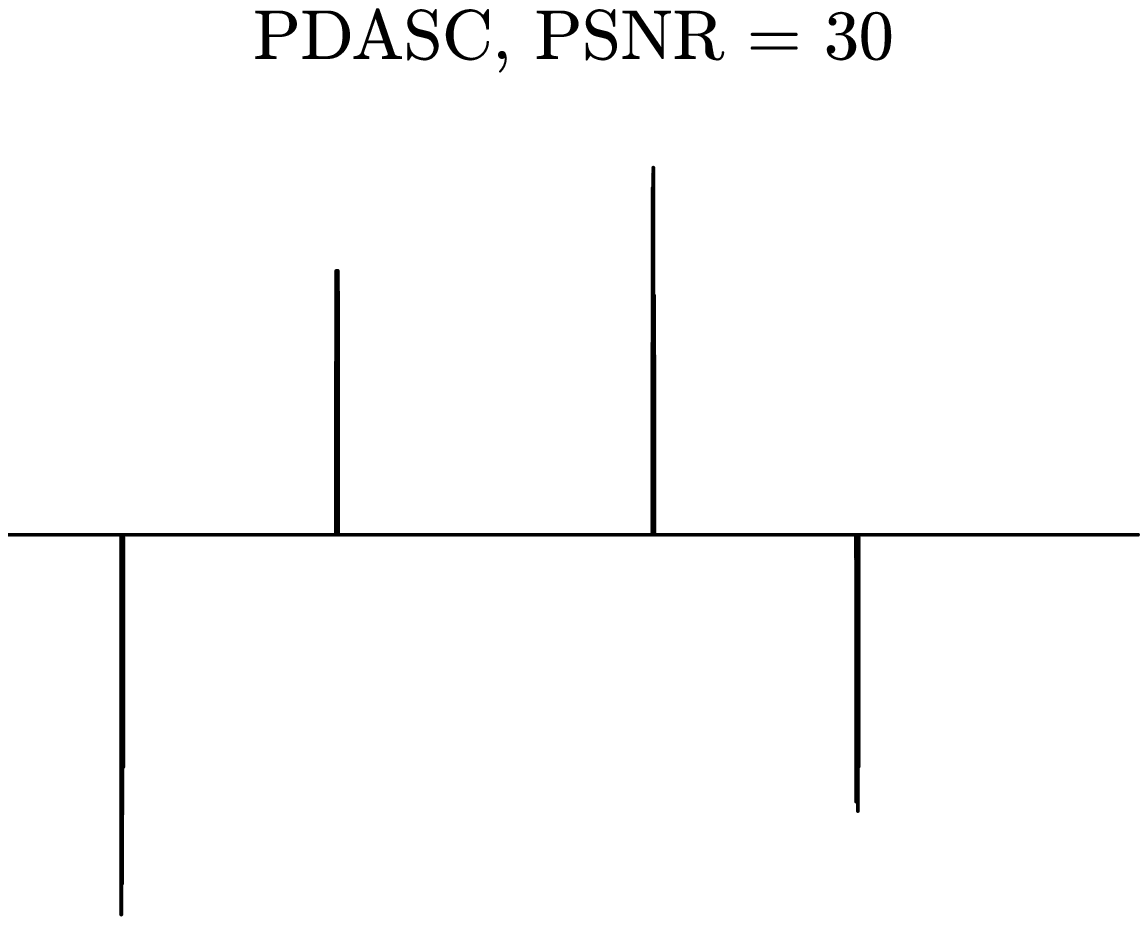}}\\
    \includegraphics[trim = 1cm 0.5cm 0.5cm 0cm, clip=true,width=3.8cm]{{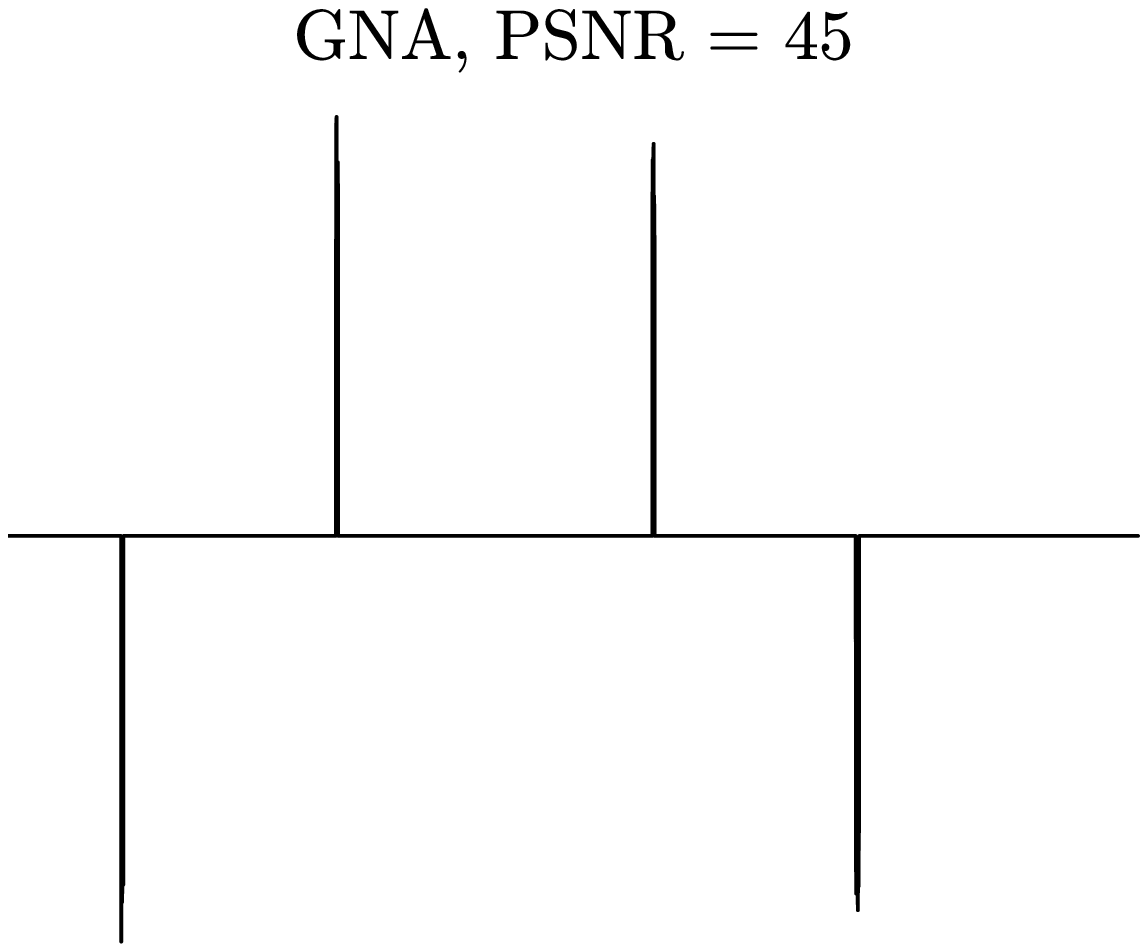}}& &
  \end{tabular}
  \caption{Reconstruction  of the one-dimensional signal with  $(m=2500, n=8000, s=36,\nu = 0, \sigma=0.5, p = 6\%)$.}\label{fig:1d}
\end{figure}

\begin{table}[ht!]
\centering
 \caption{The CPU time in seconds  and the PSNR of one-dimensional signal recovery with  $(m=2500, n=8000, s=36,\nu = 0, \sigma=0.5, p = 6\%)$.}\label{tab:1d}
 \begin{tabular}{ccccc}
 \hline
    method   &CPU time (s)  &PSNR      \\
 \hline
  BIHT         & 5.01  & 28   \\
  AOP          & 5.83  & 32   \\
  LP           & \textbf{0.12}  & 32   \\
  PBAOP        & 5.65  & 31  \\
  PDASC        &3.58   &  30\\
  GNA          &0.72   &  \textbf{45}  \\
  \hline
  \end{tabular}
\end{table}

\begin{table}[ht!]
\centering
 \caption{The CPU time in seconds  and the PSNR of two-dimensional image recovery with  $(m=5000, n=128^2, s=1138,\nu = 0, \sigma=0.05, p = 1\%)$.}\label{tab:2d}
 \begin{tabular}{ccccc}
 \hline
    method   &CPU time (s)  &PSNR      \\
 \hline
  BIHT         & 40.7 & 17   \\
  AOP          & 20.0 & 17   \\
  LP           &\textbf{ 0.21} & 17   \\
  PBAOP        & 20.5 & 16  \\
  PDASC        &21.2  & 19  \\
  GNA          &5.73  & \textbf{23}   \\
  \hline
  \end{tabular}
\end{table}

\begin{figure}[ht!]
  \centering
  \begin{tabular}{ccc}
    \includegraphics[trim = 1cm 0.5cm 0.5cm 0cm, clip=true,width=3.8cm]{{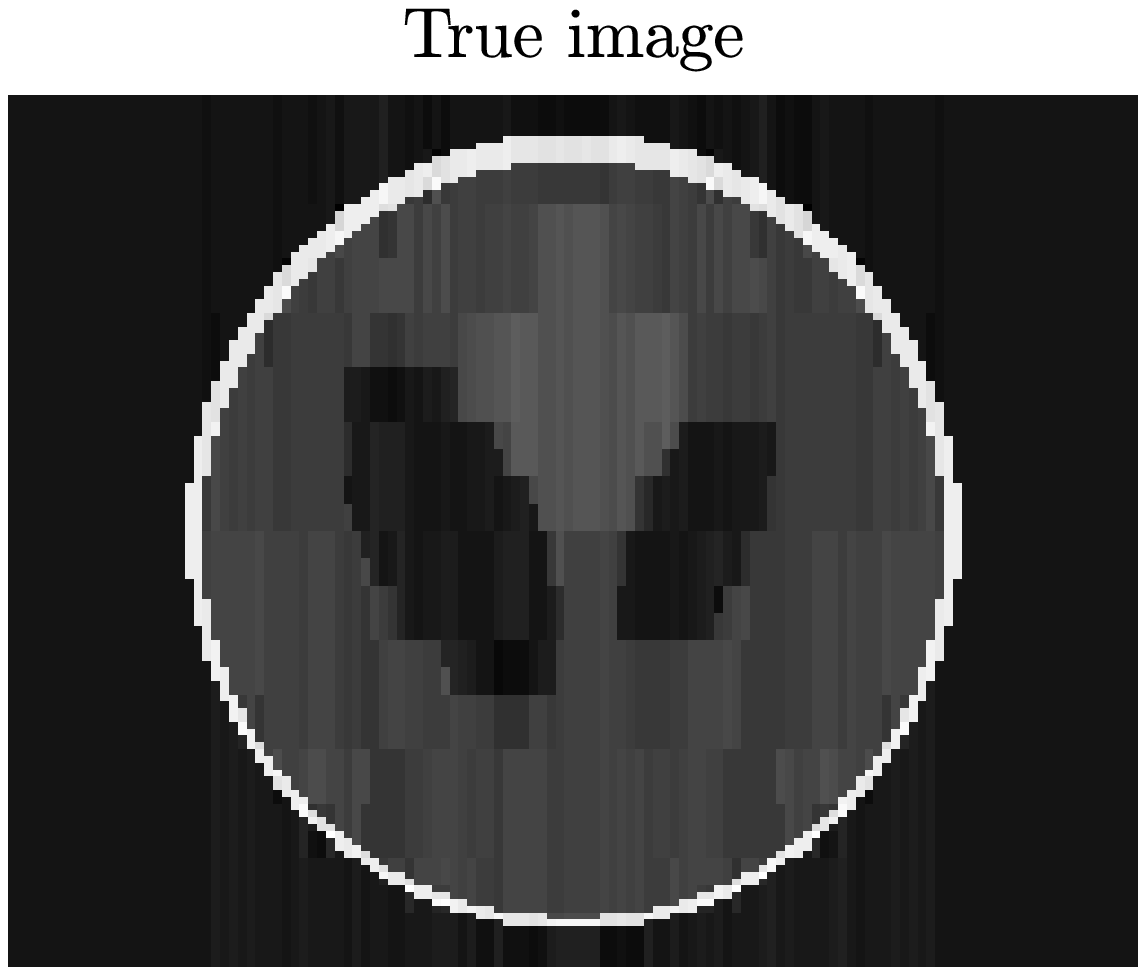}} &
    \includegraphics[trim = 1cm 0.5cm 0.5cm 0cm, clip=true,width=3.8cm]{{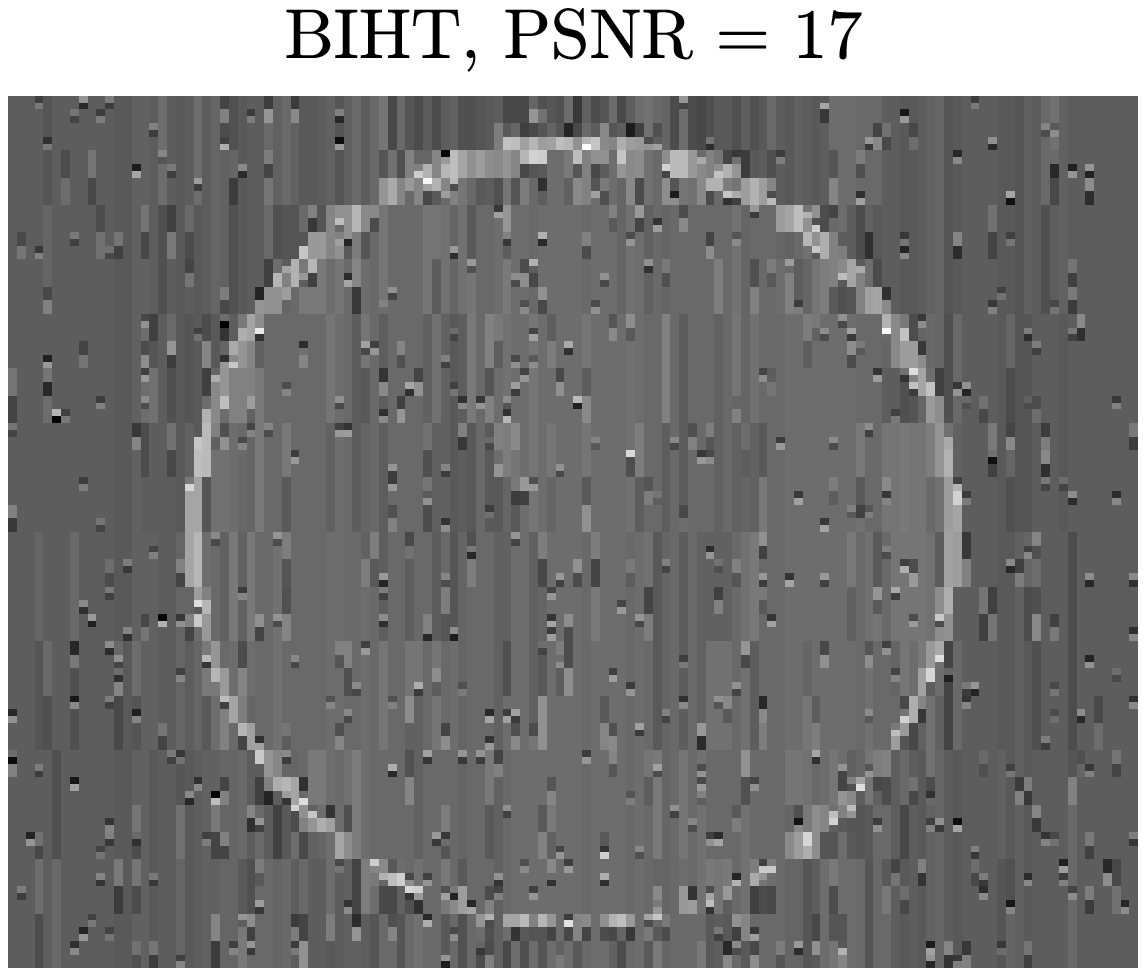}}&
    \includegraphics[trim = 1cm 0.5cm 0.5cm 0cm, clip=true,width=3.8cm]{{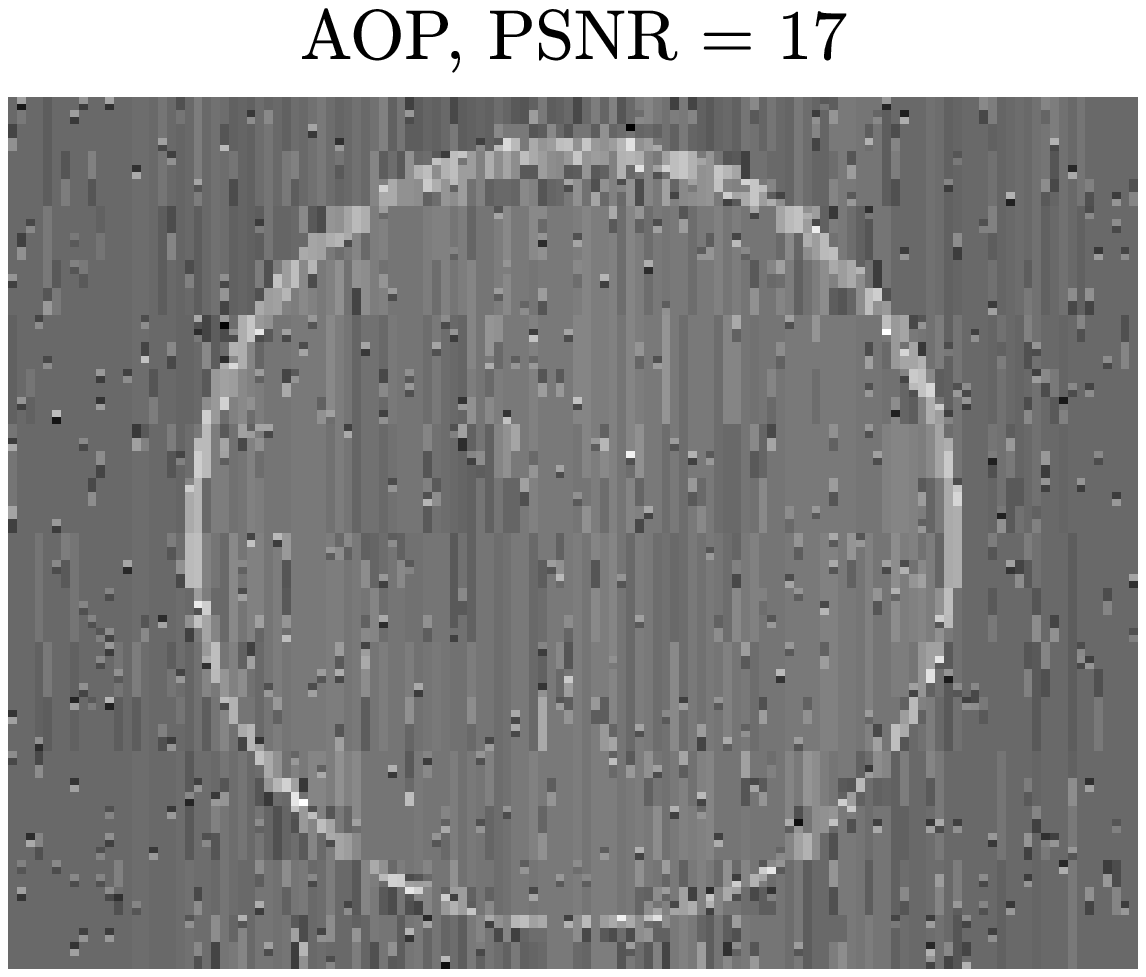}} \\
    \includegraphics[trim = 1cm 0.5cm 0.5cm 0cm, clip=true,width=3.8cm]{{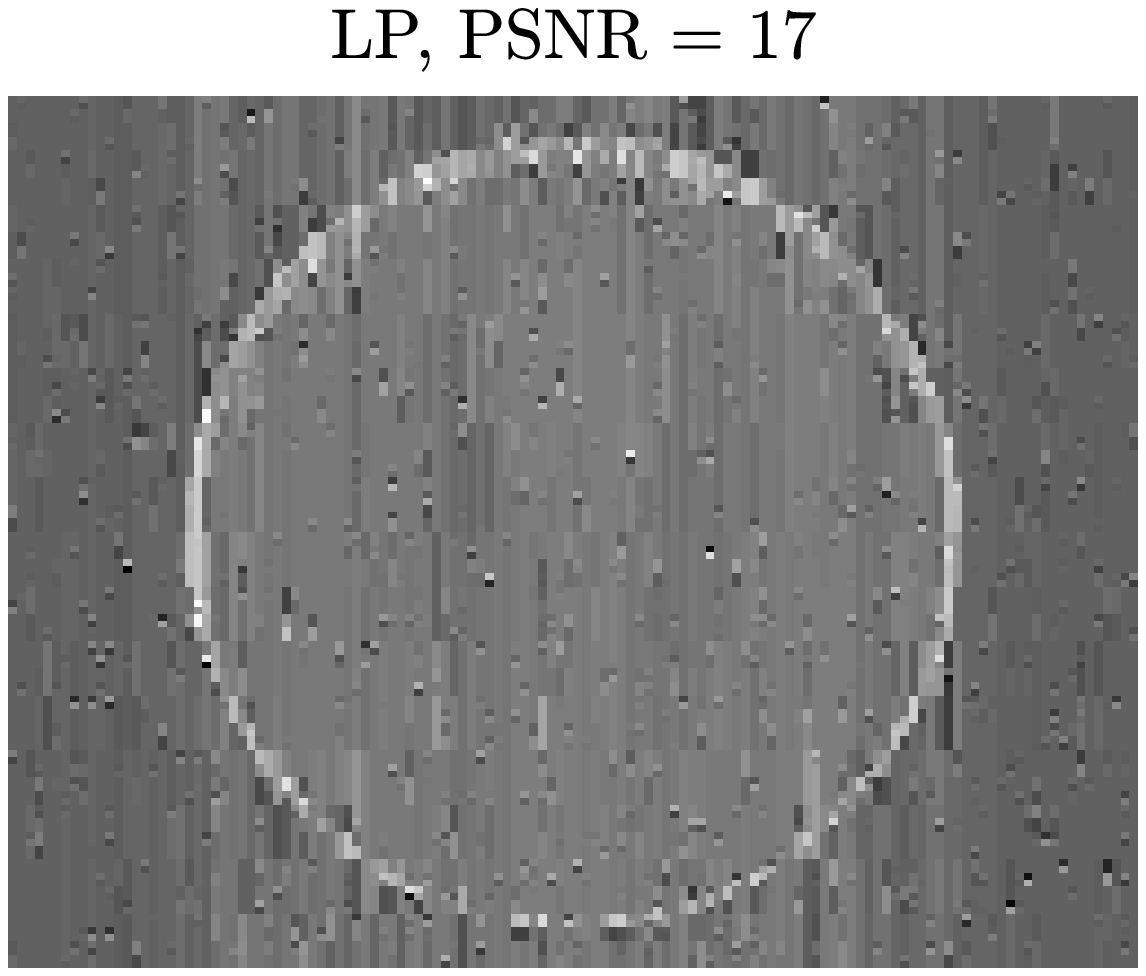}}&
    \includegraphics[trim = 1cm 0.5cm 0.5cm 0cm, clip=true,width=3.8cm]{{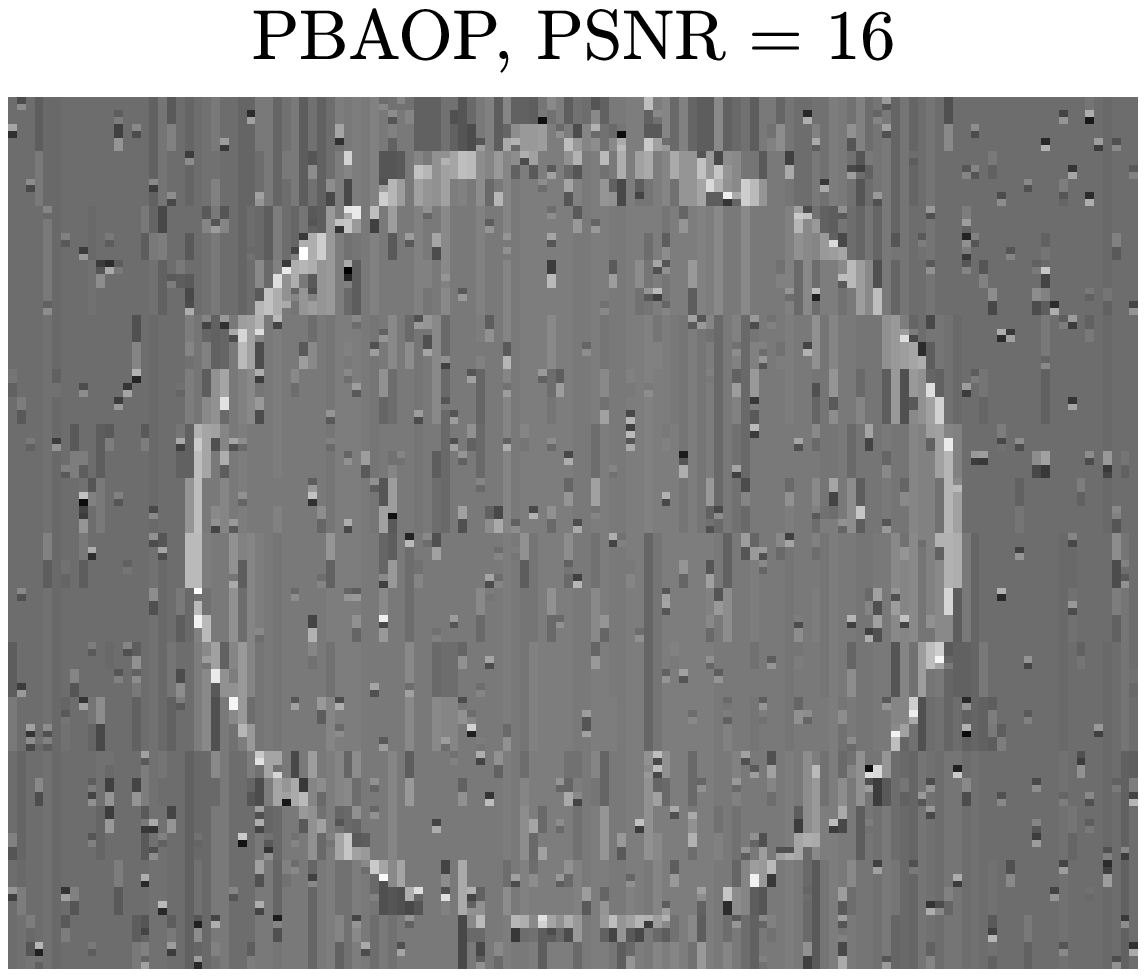}} &
    \includegraphics[trim = 1cm 0.5cm 0.5cm 0cm, clip=true,width=3.8cm]{{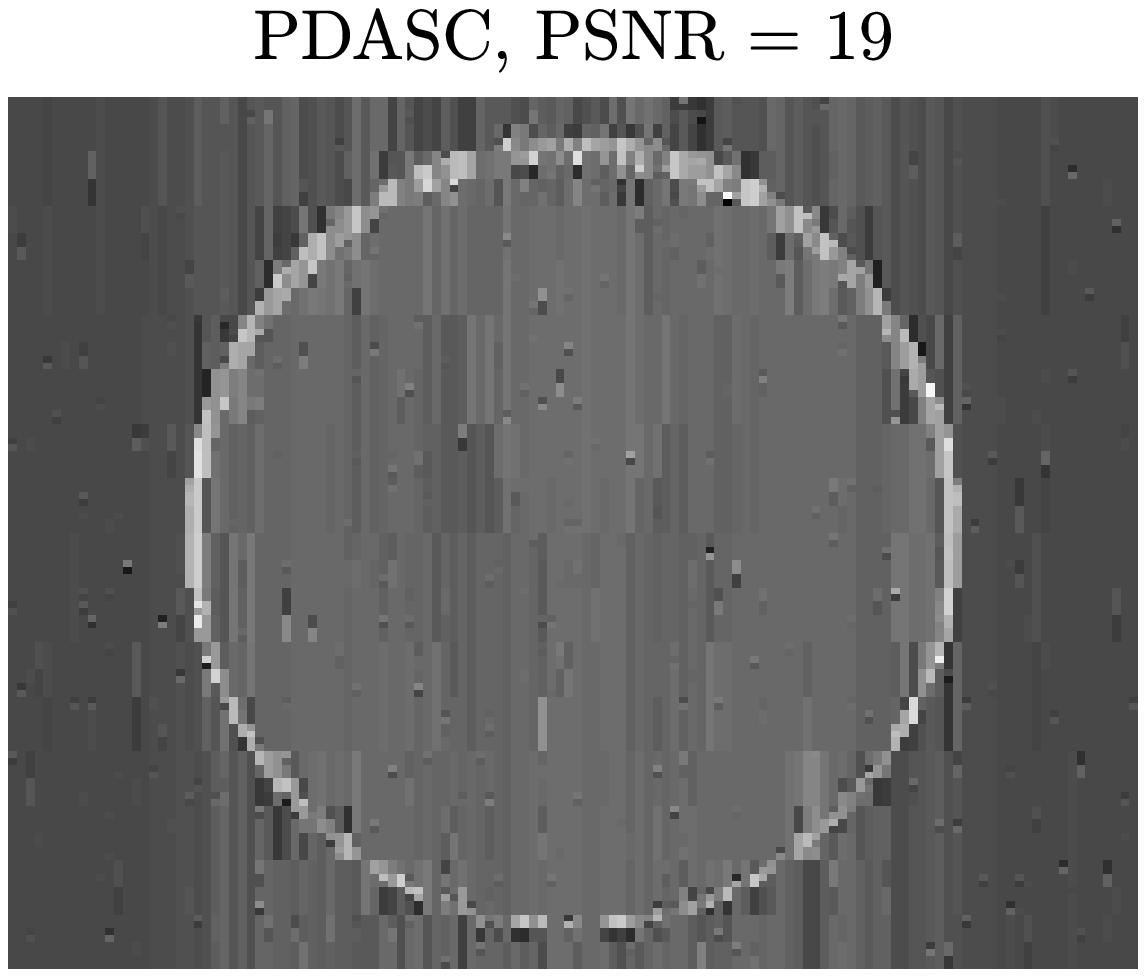}}\\
     \includegraphics[trim = 1cm 0.5cm 0.5cm 0cm, clip=true,width=3.8cm]{{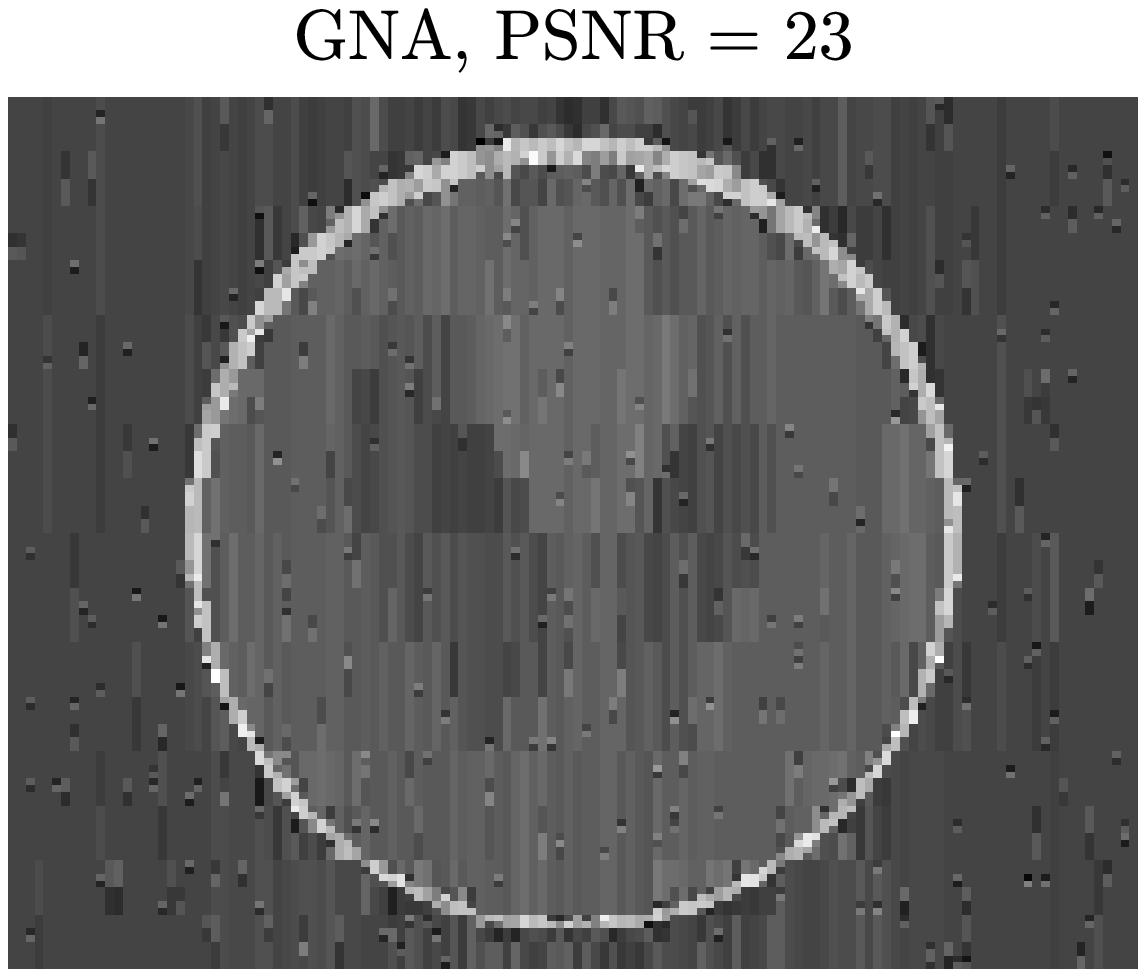}}& &
  \end{tabular}
  \caption{Reconstruction  of the two-dimension image  with  $(m=5000, n=128^2, s=1138,\nu = 0, \sigma=0.05, p = 1\%)$.}\label{fig:2d}
\end{figure}

\section{Conclusions}
 In this paper we consider decoding from binary  measurements with noise and  sign flips.
    We proposed the cardinality constraint  least squares   as a decoder.
 We prove that, up to a constant $c$, with high probability, the proposed decoder achieves a   minimax estimation error  as long as  $m \geq \mathcal{O}( s\log n)$.
 Computationally,
we utilize a generalized Newton   algorithm  (GNA) to solve  the  cardinality constraint  minimization approximately with the  cost of solving a small size least squares problem  at each iteration.
We prove that, with high probability,  the  $\ell_{\infty}$ norm of the  estimation error between  the output of GNA  and the underlying target decay to $\mathcal{O}(\sqrt{\frac{\log n }{m}}) $ with at most $\mathcal{O}(\log s)$ steps. Moreover, the underlying
support can be recover with high probability  in  $\mathcal{O}(\log s)$ steps provide that  the target signal is detectable.
Extensive numerical experiments and comparison with other state-of-the-art 1-bit CS model  demonstrates  the  robustness  of our
decoder  and  the efficiency of  GNA  algorithm.
\section*{Acknowledgements}

The work of Y. Jiao was supported in part by
the National Science Foundation of China
under Grant 11871474  and by the research fund
of KLATASDSMOE.
X.-L. Lu is partially supported by the National Science Foundation of China (No. 11871385), the National Key Research and Development Program of China (No.
2018YFC1314600) and the Natural Science Foundation of Hubei Province (No. 2019CFA007).
The research of Z. Yang   is supported by National Science Foundation of China (No. 11671312) and the Natural Science Foundation of Hubei Province (No. 2019CFA007).
This  research is supported by the Suppercomputing Center of Wuhan University.
\appendix

\section{Preliminaries}

\begin{lemma}\label{vershynin5.3940}
Let $\Psi\in \mathbb{R}^{m\times n}$  whose rows $\psi_i^t$ are independent subgaussian  vectors  in $\mathbb{R}^n$ with mean $\textbf{0}$ and covariance matrix $\Sigma$. Let $m > n$. Then for every $t>0$ with
probability at least $1-2\exp{(-C_1 t^2)}$, one has
\begin{equation}\label{norm1}
 (1- \tau)\sqrt{\gamma_{\mathrm{min}}(\Sigma)} \leq \sqrt{\gamma_{\textrm{min}}(\frac{\Psi^t\Psi}{m})}\leq \sqrt{\gamma_{\textrm{max}}(\frac{\Psi^t\Psi}{m})} \leq (1 +\tau)\sqrt{\gamma_{\mathrm{max}}(\Sigma)},
\end{equation}
and
\begin{equation}\label{norm2}
 \|\Psi^t\Psi/m - \Sigma\|\leq \max\{\tau,\tau^2\} \gamma_{\textrm{max}}(\Sigma),
 \end{equation}
 where $\tau = C_2 \sqrt{\frac{n}{m}} + \frac{t}{\sqrt{m}},$
 and  $C_1, C_2$ are generic positive constants depending   on  the
 maximum subgaussian norm of rows of $\Psi$.
\end{lemma}
\begin{proof}
Let $\Phi = \Psi \Sigma^{-\frac{1}{2}}$. Then the rows of $\Phi$ are independent sub-gaussian isotropic vectors.   \eqref{norm1} follows from Theorem 5.39 and Lemma 5.36 of \cite{Vershynin:2010} and
\eqref{norm2} is a direct consequence of  Remark 5.40 of \cite{Vershynin:2010}.
\end{proof}

%

 \begin{lemma}\label{noiselinf}
  Let $C_3 \geq \|x^*\|_1$.
If  $m> \frac{4C_1}{C_2^2}\log n$, then with probability at least $1-2/n^3  -2/n^2$, one has
  \begin{equation}
\|\Psi^t (y - c\Psi x^*)/m\|_{\infty} \leq   \frac{2(1+|c|C_3)}{\sqrt{C_1}}\sqrt{\frac{\log n}{m}}.
  \end{equation}
  \end{lemma}

  \begin{proof}
  See Lemma D.2 in \cite{HuangJiao:2018}.
  \end{proof}

\begin{lemma}\label{lb}
Define
\begin{align*}
&C_{2s,\mathrm{min}} = \inf_{A\subset [n], |A|  \leq 2s } \frac{\gamma_{\mathrm{min}}(\Psi_{A}^t\Psi_{A})}{m},\\
 &C_{2s,\mathrm{max}} = \sup_{A\subset [n], |A|  \leq 2s } \frac{\gamma_{\mathrm{max}}(\Psi_{A}^t\Psi_{A})}{m}.
 \end{align*}
Then, with probability at least $ 1-  4/n^2$,
\begin{align*}
&C_{2s,\mathrm{min}} \geq  \gamma_{\mathrm{min}}(\Sigma)/4,\\
&C_{2s,\mathrm{max}} \leq   9\gamma_{\mathrm{max}}(\Sigma)/4,
\end{align*}
 as long as $m \geq \frac{16(C_2+1)^2}{C_1}s\log\frac{en}{s}$ and $s\leq \exp^{(1-\frac{C_1}{2})}n.$
\end{lemma}
\begin{proof}
Given $A\subset [n], |A|  \leq 2s  $, we define the event $$E_{A} = \{\sqrt{\frac{\gamma_{\mathrm{min}}(\Psi_{A}^t\Psi_{A})}{m}} > \sqrt{\gamma_{\mathrm{min}}(\Sigma)}(1-C_{2} \sqrt{\frac{2s}{m}}-\frac{t}{\sqrt{m}})\}.$$ Then,
\begin{align}
\mathbb{P}[C_{2s,\mathrm{min}} > \gamma_{\mathrm{min}}(\Sigma)(1-C_{2} &\sqrt{\frac{2s}{m}}-\frac{t}{\sqrt{m}})^2] = \mathbb{P}[\bigcap_{A\in [n], |A|\leq 2s} E_{A}]\nonumber = \mathbb{P}[\bigcap_{A\in [n], |A|= \ell, 1\leq\ell\leq 2s} E_{A}]\nonumber\\
& = 1 - \mathbb{P}[\bigcup_{A\in [n], |A|= \ell, 1\leq\ell\leq 2s} \overline{E_{A}}]\nonumber\\
&\geq 1- \sum_{\ell = 1}^{2s}\sum_{A\subset[n], |A| =  \ell} (1-\mathbb{P}[E_{A}])\nonumber\\
&\geq 1- \sum_{\ell = 1}^{2s}\sum_{A\subset [n], |A|\leq \ell}  2\exp(-C_1 t^2)\nonumber\\
&= 1- \sum_{\ell = 1}^{2s}\binom{n}{\ell}2 \exp(-C_1 t^2 )\nonumber\\
&\geq 1- 2(\frac{en}{2s})^{2s} \exp(-C_1 t^2 ),\nonumber
\end{align}
where the first inequality follows from the union bound, the second inequality  follows from \eqref{norm1} by replacing $\Psi$ with $\Psi_{A}$, and  the third inequality  holds since
 \begin{align*}
 &\sum_{\ell = 1}^{2s}\binom{n}{\ell}\leq (\frac{n}{2s})^{2s}\sum_{\ell =0}^{2s}\binom{n}{\ell}(\frac{2s}{n})^{\ell}\\
 &\leq (\frac{n}{2s})^{2s}(1+\frac{2s}{n})^{n}\\
 & \leq (\frac{e n}{2s})^{2s}.
 \end{align*}
Thus, by setting $t = \sqrt{\frac{4s}{C_1} \log\frac{en}{s}}$, we derive with probability at least $1- 2/(\frac{en}{s})^{2s} \geq 1-  2/n^2$
\begin{align*}
C_{2s,\mathrm{min}} &> \gamma_{\mathrm{min}}(\Sigma)(1-C_{2} \sqrt{\frac{2s}{m}}-\frac{t}{\sqrt{m}})^2\\
& \geq \gamma_{\mathrm{min}}(\Sigma)(1-(C_{2}+1)\sqrt{\frac{4s}{mC_1} \log\frac{en}{s}})^2 \\
&\geq \gamma_{\mathrm{min}}(\Sigma)/4,
\end{align*}
where the second inequality follows from  the assumption $s\leq \exp^{(1-\frac{C_1}{2})}n$ and , i.e., $$\sqrt{\frac{2s}{m}}\leq \sqrt{\frac{4s}{mC_1} \log\frac{en}{s}}<1,$$ and the last inequality follows from the assumption  $m \geq \frac{16(C_2+1)^2}{C_1}s\log\frac{en}{s}$.
Define event $$\widetilde{E}_{A} = \{\sqrt{\frac{\gamma_{\mathrm{max}}(\Psi_{A}^t\Psi_{A})}{m}} > \sqrt{\gamma_{\mathrm{max}}(\Sigma)}(1+C_{2} \sqrt{\frac{2s}{m}}+\frac{t}{\sqrt{m}})\}.$$
Then,
\begin{align}
\mathbb{P}[C_{2s,\mathrm{max}} > \gamma_{max}(\Sigma)(1+C_{2} \sqrt{\frac{2s}{m}}+\frac{t}{\sqrt{m}})^2] &= \mathbb{P}[\bigcup_{A\subset [n], |A|\leq 2s} \widetilde{E}_{A}]\nonumber\\
& = \mathbb{P}[\bigcup_{A\subset [n], |A| = \ell,   1\leq\ell\leq 2s} \widetilde{E}_{A}]\nonumber\\
& \leq \sum_{\ell = 1}^{2s}\sum_{A\in [n], |A|\leq \ell}  2 \exp(-C_1 t^2)\nonumber\\
&\leq (\sum_{\ell = 1}^{2s}\binom{n}{\ell}) 2\exp(-C_1 t^2 )\nonumber\\
&\leq  2(\frac{en}{s})^{2s} \exp(-C_1 t^2 ),\nonumber
\end{align}
which implies  with probability at least $1-2(\frac{en}{s})^{2s} \exp(-C_1 t^2 )$,
\begin{equation*}
C_{2s,\mathrm{max}} \leq  \gamma_{max}(\Sigma)(1+C_{2} \sqrt{\frac{2s}{m}}+\frac{t}{\sqrt{m}})^2.
\end{equation*}
We finish the proof by  setting $t = \sqrt{\frac{4s}{C_1} \log\frac{en}{s}}$ and some algebra.
\end{proof}

\section{Proof of Theorem \ref{errsub}}\label{app:errsub}
\begin{proof}
Let
\begin{align}
 \tilde{x}^* &= c x^*,\label{xstar}\\
 R &=  y -\Psi  \tilde{x}^*,\label{effniose}
 \end{align}
and  $ \Delta=  x_{\ell_0} - \tilde{x}^* $, $\mathcal{A} = \mathrm{supp}(R)$.
By the definition   $x_{\ell_0}$,  we have
$\|x_{\ell_0}\|_0 \leq s,$ $|\mathcal{A}|\leq 2s$
and
\begin{equation}\label{optmin}
\frac{1}{2m}\|y - \Psi x_{\ell_0}\|_2^2 \leq \frac{1}{2m}\|y - \Psi \tilde{x}^*\|_2^2.
\end{equation}
Then,
\begin{align*}
\gamma_{\mathrm{min}}(\Sigma)\|\Delta\|_2^2/4 &\leq C_{2s,\mathrm{min}}\|\Delta\|_2^2 \leq \frac{1}{2m}\|\Psi \Delta\|_2^2 \\
   &\leq \langle \Delta, \Psi^t R/m \rangle\leq \|\Delta\|_2 \|\Psi_{\mathcal{A}}^t R/m\|_2\\
 &\leq \sqrt{2s} \|\Delta\|_2 \| \Psi^t R/m\|_{\infty}\\
 & \leq \|\Delta\|_2  \frac{2(1+|c|C_3)}{\sqrt{C_1}}\sqrt{\frac{2s\log n}{m}},
 \end{align*}
 where, the first inequality holds with probability at least $1-2/n^2$ by  Lemma \ref{lb}, and  the second inequality uses the definition of   $C_{2s,\mathrm{min}}$, and the third  inequality dues to  \eqref{optmin} and some algebra, and the third on uses  Cauchy Schwartz inequality, and the fourth inequity follow from $|\mathcal{A}|\leq 2s$ and Cauchy Schwartz inequality,   and the last inequality holds with probability  at least $1-2/n^3  -2/n^2$  by Lemma \ref{noiselinf}.
 The above display implies, with probability at least $1-2/n^3  -4/n^2$
  $$\|x_{\ell_0}/c - x^*\| \leq  \frac{12(1/|c|+C_3)}{\sqrt{C_1} \gamma_{\mathrm{min}}(\Sigma)}\sqrt{\frac{s\log n}{m}}.$$
  We finish the proof by substituting $c$ and some algebra.
\end{proof}

\section{Proof of Lemma  \ref{kkt}}\label{app:kkt}
\begin{proof}
Flowing Theorem 2.2 in \cite{Beck:2013}, we just need to show that $\Psi$ is $s$-regular i.e., $\frac{\gamma_{\mathrm{min}}(\Psi_{A}^t\Psi_{A})}{m}>0$ with $A\subset [n], |A|\leq s$,   and to calculate  the Lipschitz constant of the gradient of the least squares  loss in \eqref{setup} restricted on $s$ sparse vectors, i.e., $\frac{\gamma_{\mathrm{max}}(\Psi_{A}^t\Psi_{A})}{m}$  with $A\subset [n], |A|\leq 2s$.
By Lemma \ref{lb} with probability at least $1-4/n^2$,  $\Psi$ is $s$-regular and the Lipschitz constant is bounded by $9\gamma_{\mathrm{max}}(\Sigma)/4$.
\end{proof}

\section{Proof of Proposition \ref{eqn}}\label{app:eqn}
\begin{proof}
Denote $D^k = - (H^k)^{-1}F(w^k)$. Then,
$$ w^{k+1} = w^k - (H^k)^{-1}F(w^k),$$ can be recast as
\begin{align}
H^k D^k &= -F(w^k)\label{newton1}\\
w^{k+1} &= w^k + D^k\label{newton2}.
\end{align}
  Partition $w^k$,   $D^k$ and  $F(w^k)$ according to  $A^{k}$ and ${I}^{k}$ such that
  \begin{equation} \label{FF2}
       w^{k}=\left(
         \begin{array}{c}
         x_{A^k}^{k} \\
           x_{I^{k}}^k \\
           d_{{A}^{k}}^k \\
          d_{{I}^{k}}^k \\
         \end{array}
          \right),\ \
       D^{k}=\left(
         \begin{array}{c}
         D^{x}_{A^{k}} \\
          D^{x}_{I^{k}} \\
           D^{d}_{A^{k}} \\
          D^{d}_{{I}^{k}} \\
         \end{array}
          \right),
          \end{equation}
       \begin{equation}\label{FF3}
F(w^{k})=
\left[\begin{array}{c}
 - d_{{A}^{k}}^{k} \\
 x_{{I}^{k}}^{k}  \\
  \Psi^{t}_{{A}^{k}} \Psi_{{A}^{k}} x_{{A}^{k}}^{k} +  \Psi^t_{{A}^{k}} \Psi_{{I}^{k}} x_{{I}^{k}}^{k}+ m d_{{A}^{k}}^{k} - \Psi^t_{{A}^{k}}y     \\
\Psi^{t}_{{I}^{k}} \Psi_{\mathcal{A}^{k}} x_{\mathcal{A}^{k}}^{k} +  \Psi^{t}_{{I}^{k}} \Psi_{{I}^{k}} x_{{I}^{k}}^{k}+ m d_{{I}^{k}}^{k} - \Psi^t_{{I}^{k}}y
\end{array}\right].
\end{equation}
 Substituting \eqref{FF2}-\eqref{FF3} and $H^k$ into \eqref{newton1}, we have
\begin{align}
  (d_{{A}^{k}}^{k} + D^{d}_{\mathcal{A}_{k}})    &= \mathbf{0}_{A^k},\label{eqv1}\\
   x^{k}_{{I}^{k}} + D^{x}_{{I}^{k}}  &= \textbf{0}_{I^k},\label{eqv2}\\
\Psi^{t}_{{A}^{k}} \Psi_{{A}^{k}} (x_{{A}^{k}}^{k} + D^{x}_{{A}^{k}})&= \Psi^t_{{A}^{k}}y  - m( d_{{A}^{k}}^{k} + D^{d}_{{A}^{k}})- \Psi^t_{{A}^{k}} \Psi_{{I}^{k}} ( x^{k}_{{I}^{k}} + D^{x}_{{I}^{k}}),  \label{eqv3}\\
m(d_{{I}^{k}}^{k} + D^{d}_{{I}^{k}}) &= \Psi^t_{{I}^{k}}y - \Psi^{t}_{{I}^{k}} \Psi_{{A}^{k}}(x_{{A}^{k}}^{k} + D^{x}_{{A}^{k}}) - \Psi^t_{{A}^{k}} \Psi_{{I}^{k}}(x_{{I}^{k}}^{k} + D^{x}_{{I}^{k}}) . \label{eqv4}
\end{align}
It follows from \eqref{newton2} that
 \begin{equation}\label{relation}
 \left( \begin{array}{c}
  x_{{A}^{k}}^{k+1} \\
          x_{{I}^{k}}^{k+1} \\
           d_{{A}^{k}}^{k+1} \\
          d_{{I}^{k}}^{k+1} \\
         \end{array}
       \right)
       = \left(
         \begin{array}{c}
           x_{{I}^{k}}^{k}+D^{x}_{{I}^{k}} \\
           d_{{A}^{k}}^{k} +D^{d}_{{A}^{k}}\\
           x_{{A}^{k}}^{k}+D^{x}_{{A}^{k}} \\
          d_{{I}^{k}}^{k}+ D^{d}_{{I}^{k}} \\
         \end{array}
       \right).
 \end{equation}
Substituting  \eqref{relation} into \eqref{eqv1} - \eqref{eqv4}, we get  \eqref{eq5} of Algorithm  \ref{alg:genew}.
This completes the  proof.
\end{proof}

\section{Proof of Theorem \ref{th:est}}\label{app:est}
Let $A^* = \mathrm{supp}(x^*)$ and $F(x) = \|y-\Psi x\|^2_{2}/m$  be the least squares loss in \eqref{subreg}.  And recall $
 \tilde{x}^* = c x^*$ in \eqref{xstar} and $R =  y -\Psi  \tilde{x}^*$ in \eqref{effniose}.
 By \eqref{eq5}, it is easy to see that $$d^k = -\nabla F(x^k), \ \ \langle d^k, x^k\rangle = 0, \ \ k \geq 1.$$
The proof of Theorem \ref{th:est} is based on the following Lemmas \ref{L9.1}-\ref{L9.3}, whose proof are shown in Appendix \ref{plemma}.
\begin{lemma}\label{L9.1}
Let $\mathcal{A}^k = A^{k}\backslash A^{k-1}$ and $\varrho_{k} =\frac{|\mathcal{A}^k|}{|\mathcal{A}^k|+|A^*\backslash A^{k-1}|}$, $k\geq1.$
\begin{equation*}
2C_*\varrho_{k} (F(x^k)-F(\tilde{x}^*)) \leq  \|d^k_{\mathcal{A}^k}\|_1\|d^k_{\mathcal{A}^k}\|_{\infty}.
\end{equation*}
\end{lemma}
\begin{lemma}\label{L9.2}
Let $\zeta =1-\frac{2\eta C_*(1-\eta \sqrt{s}C^*)}{ \sqrt{s}(1+s)}\in(0,1).$
 It holds
\begin{equation*}
F(x^{k+1})-F(\tilde{x}^*)\leq\zeta (F(x^k)-F(\tilde{x}^*)),
\end{equation*}
before Algorithm GNA  terminates.
\end{lemma}
\begin{lemma}\label{L9.3}
Let $\eta \in (0, \frac{1}{C^*\sqrt{s}})$ and  $x^0=\mathbf{0}$ in  GNA.
Then
we have
\begin{align}\label{ester}
\|x^k/c-x^*\|_{\infty}\leq
\zeta^{k/2}(\sqrt{\frac{2\|\Psi^tR/m\|_{\infty}\|x^*\|_1}{C_*|c|}}+\|x^*\|_1\sqrt{\frac{C^*}{C_*}})+\frac{2\|\Psi^tR/m\|_{\infty}}{C_*|c|}.
\end{align}
\end{lemma}
\begin{proof}
By definition  $C^* \leq C_{2s,\mathrm{max}}$. Then by Lemma \ref{lb}, the step size
$\eta \in (0,\frac{4}{9\gamma_{\mathrm{max}}(\Sigma)\sqrt{s}})$ satisfying Lemma \ref{L9.3} with probability at least $1-4/n^2$.
By the assumption  $\|x^*\|_{\Sigma} = 1$ and Cauchy Schwartz inequality,   we get
\begin{equation}\label{l1}
\|x^*\|_1 \leq \sqrt{\frac{s}{\gamma_{\mathrm{min}}(\Sigma)}}.
\end{equation}
By Lemma \ref{noiselinf},
\begin{equation}\label{nlinf}
\|\Psi^t R/m\|_{\infty} \leq   \frac{2(1+|c|C_3)}{\sqrt{C_1}}\sqrt{\frac{\log n}{m}}
\end{equation}
holds with probability at least $1-2/n^3  -2/n^2$.
Substituting \eqref{l1} and  \eqref{nlinf} into \eqref{ester} and  some algebra completes the proof.
\end{proof}

\section{Proof of Lemmas \ref{L9.1}-\ref{L9.3}}\label{plemma}
\subsection{Proof of Lemma \ref{L9.1}}

\begin{proof}
 In the scenario $A^{k}=A^{k-1}$ or $F(x^k)\leq F(\tilde{x}^*)$, the desired result holds trivially.  Therefore, we assume  $A^{k}\neq A^{k-1}$ and $F(x^k)>F(\tilde{x}^*)$. By the definition of $C_{*}, C^*$ in \eqref{CC} and Taylor expansion we have,
$$
\frac{C_*}{2} \|\tilde{x}^*-x^k\|_1\|\tilde{x}^*-x^k\|_{\infty} \leq F(\tilde{x}^{*})-F(x^k)+\langle d^k,{\tilde{x}^*-x^k}\rangle
$$
The above display implies,
\begin{align*}
&\langle d^k,\tilde{x}^*\rangle = \langle d^k,{\tilde{x}^*-x^k}\rangle\\
&\geq\frac{C_*}{2}\|\tilde{x}^*-x^k\|_1\|\tilde{x}^*-x^k\|_{\infty}+F(x^k)-F(\tilde{x}^{*})\\
&\geq \sqrt{2C_*}\sqrt{\|\tilde{x}^*-x^k\|_1\|\tilde{x}^*-x^k\|_{\infty}}\sqrt{F(x^k)-F(\tilde{x}^{*})}.
\end{align*}
By the definition of $A^{k}$ in \eqref{eq4} and $x^k,d^k$ in \eqref{eq5},  $\mathcal{A}^k$ contains the first $|\mathcal{A}^k|$-largest elements in absolute value of $d^k$ and
$\mathrm{supp}(d^k)\bigcap \mathrm{supp}(\tilde{x}^*) = A^{*}\backslash A^{k-1}.$
By the definition of   $\varrho_k$ and Cauchy Schwartz inequality, we have
\begin{align*}
&\langle d^k,\tilde{x}^*\rangle \leq\frac{1}{\sqrt{\varrho_{k}}}\|d^k_{\mathcal{A}^k}\|_2\|\tilde{x}_{A^*\backslash A^{k-1}}^{*}\|_2\\
&=\frac{1}{\sqrt{\varrho_{k}}}\|d^k_{\mathcal{A}^k}\|_2\|(\tilde{x}^{*}-x^k)_{A^*\backslash A^{k-1}}\|_2\\
&\leq\frac{1}{\sqrt{\varrho_{k}}}\sqrt{\|d^k_{\mathcal{A}^k}\|_1\|d^k_{\mathcal{A}^k}\|_{\infty}}\sqrt{\|\tilde{x}^{*}-x^k\|_1\|\tilde{x}^{*}-x^k\|_{\infty}}.
\end{align*}
We finish the proof by combing the above two displays.
\end{proof}

\subsection{Proof of Lemma \ref{L9.2}}

\begin{proof}
Let $u^{k}= x^k+\eta d^k,\ \  k \geq 1$.
By the definition of $u^k$ and $A^{k}$ in \eqref{eq4} and $x^k,d^k$ in \eqref{eq5}, we have
\begin{align*}
&\langle -d^{k+1}, u^{k+1}|_{A^{k+1}}-x^{k+1}\rangle =\langle-d^{k+1}, u^{k+1}|_{A^{k+1}}\rangle =\langle - d^{k+1}_{A^{k+1}\backslash A^{k}}, u^{k+1}_{A^{k+1}\backslash A^{k}}\rangle,\\
&|A^{k}\backslash A^{k+1}|=|A^{k+1}\backslash A^{k}|, \ \  u_{A^{k}\backslash A^{k+1}}^{k+1}=x_{A^{k}\backslash A^{k+1}}^{k+1}, \ \   u_{A^{k+1}\backslash A^{k}}^{k+1} = \eta d^{k+1}_{A^{k+1}\backslash A^{k}},\\
&\|u_{A^{k}\backslash A^{k+1}}^{k+1}\|_1=\|x_{A^{k}\backslash A^{k+1}}^{k+1}\|_1 \leq \|u_{A^{k+1}\backslash A^{k}}^{k+1}\|_1,  \ \
\max\{\|u_{A^{k+1}\backslash A^{k}}^{k+1}\|_{\infty}, \|x_{A^{k}\backslash A^{k+1}}^{k+1}\|_{\infty}\}= \|u_{A^{k+1}\backslash A^{k}}^{k+1}\|_{\infty}.
\end{align*}
Then,
\begin{align}\label{ul1}
&\|u^{k+1}|_{A^{k+1}}-x^{k+1}\|_1= \|u^{k+1}|_{A^{k+1}\backslash A^{k}}+u^{k+1}|_{A^{k+1}\bigcap A^{k}}-x^{k+1}|_{A^{k+1}\bigcap A^{k}}-x^{k+1}|_{A^{k}\backslash A^{k+1}}\|_1\nonumber\\
&=\|u_{A^{k+1}\backslash A^{k}}^{k+1}\|_1 +\|u_{A^{k+1}\bigcap A^{k}}^{k+1}-x_{A^{k+1}\bigcap A^{k}}^{k+1}\|_1+ \|x_{A^{k}\backslash A^{k+1}}^{k+1}\|_1\nonumber\\
&=\|u_{A^{k+1}\backslash A^{k}}^{k+1}\|_1+\|x_{A^{k}\backslash A^{k+1}}^{k+1}\|_1 \leq 2\|u_{A^{k+1}\backslash A^{k}}^{k+1}\|_1 = 2\eta \|d_{A^{k+1}\backslash A^{k}}^{k+1}\|_1,
\end{align}
\begin{align}\label{ulinf}
&\|u^{k+1}|_{A^{k+1}}-x^{k+1}\|_{\infty} = \|u_{A^{k+1}\backslash A^{k}}^{k+1}\|_{\infty}+\|x_{A^{k}\backslash A^{k+1}}^{k+1}\|_{\infty}\nonumber\\
&=\max\{\|u_{A^{k+1}\backslash A^{k}}^{k+1}\|_{\infty}, \|x_{A^{k}\backslash A^{k+1}}^{k+1}\|_{\infty}\} =\|u_{A^{k+1}\backslash A^{k}}^{k+1}\|_{\infty} = \eta \|d_{A^{k+1}\backslash A^{k}}^{k+1}\|_{\infty}.
\end{align}
By  \eqref{ul1}-\eqref{ulinf} and  the  definition of $C^*$ in \eqref{CC} and Taylor expansion, we get
\begin{align*}
&F(u^{k+1}|_{A^{k+1}})-F(x^{k+1})\\
&\leq \langle - d^{k+1}, u^{k+1}|_{A^{k+1}}-x^{k+1}\rangle + \frac{C^*}{2}\|u^{k+1}|_{A^{k+1}}-x^{k+1}\|_1\|u^{k+1}|_{A^{k+1}}-x^{k+1}\|_{\infty}\\
&\leq \langle - d^{k+1}_{A^{k+1}\backslash A^{k}}, u^{k+1}_{A^{k+1}\backslash A^{k}}\rangle + \frac{C^*}{2} 2\eta \|d_{A^{k+1}\backslash A^{k}}^{k+1}\|_1\eta \|d_{A^{k+1}\backslash A^{k}}^{k+1}\|_{\infty}\\
& \leq  -\eta\|d_{A^{k+1}\backslash A^{k}}^{k+1}\|_2^2 +\eta^2C^*\|d_{A^{k+1}\backslash A^{k}}^{k+1}\|_1 \|d_{A^{k+1}\backslash A^{k}}^{k+1}\|_{\infty}\\
& \leq (-\frac{\eta}{\sqrt{s}} + \eta^2C^*)\|d_{A^{k+1}\backslash A^{k}}^{k+1}\|_1 \|d_{A^{k+1}\backslash A^{k}}^{k+1}\|_{\infty}.
\end{align*}
Then the by the definition of $x^{k+1}$ and the above display, we deduce,
 \begin{align*}
  & F(x^{k+1})- F(\tilde{x}^*) - (F(x^{k}) - F(\tilde{x}^*))\leq F(u^{k}|_{A^{k}})-F(x^{k})\\
& \leq (-\frac{\eta}{\sqrt{s}} + \eta^2C^*)\|d_{A^{k}\backslash A^{k-1}}^{k}\|_1 \|d_{A^{k}\backslash A^{k-1}}^{k}\|_{\infty}\\
& \leq (-\frac{\eta}{\sqrt{s}} + \eta^2C^*)2C_*\varrho_{k} (F(x^k)-F(\tilde{x}^*))\\
& \leq (-\frac{\eta}{\sqrt{s}} + \eta^2C^*)2C_*\frac{1}{1+s} (F(x^k)-F(\tilde{x}^*))
\end{align*}
where the third inequality  uses
$\eta <\frac{1}{\sqrt{s}C^*}$ and Lemma \ref{L9.1} and the fourth inequality holds due to $\varrho_{k} \geq \frac{1}{s+1}$.
We finish the proof by rearranging term in the above display.
\end{proof}
\section{Proof of Lemma \ref{L9.3}}

\begin{proof}
If $\|x^k-\tilde{x}^*\|_{\infty}< \frac{2\|\Psi^tR/m\|_{\infty}}{C_*}$, Lemma \ref{L9.3} holds trivially. Therefore,  we consider the case $$\|x^k-\tilde{x}^*\|_{\infty}\geq \frac{2\|\Psi^tR/m\|_{\infty}}{C_*}.$$
It follows from the  the definition of $C_*$ and Taylor expansion that
\begin{align*}
&F(x^{k})-F(\tilde{x}^{*}) \geq \langle \nabla F(\tilde{x}^{*}),x^{k}-\tilde{x}^{*}\rangle+\frac{C_*}{2}\|x^{k}-\tilde{x}^{*}\|_1\|x^{k}-\tilde{x}^{*}\|_{\infty}\\
&\geq-\|\Psi^tR/m\|_{\infty}\|x^{k}-\tilde{x}^{*}\|_1+\frac{C_*}{2}\|x^{k}-\tilde{x}^{*}\|_1\|x^{k}-\tilde{x}^{*}\|_{\infty}.
\end{align*}
The above display and the fact
$$
(\|x^{k}-\tilde{x}^{*}\|_1-\|x^{k}-\tilde{x}^{*}\|_{\infty})(\frac{C_*}{2}\|x^{k}-\tilde{x}^{*}\|_{\infty}-\|\Psi^tR/m\|_{\infty})\geq 0,
$$
imply
$g(\|x^{k}-\tilde{x}^{*}\|_{\infty}) \leq 0$,
where the univariate quadratic function $$g(t) = \frac{C_*}{2}t^2 - \|\Psi^tR/m\|_{\infty}  t -(F(x^k)-F(\tilde{x}^*)).$$
Therefore, we have
\begin{align}\label{apeq1}
\|x^{k}-\tilde{x}^{*}\|_{\infty} &\leq\frac{\|\Psi^tR/m\|_{\infty}+\sqrt{\|\Psi^tR/m\|_{\infty}^2+2C_*(F(x^{k})-F(\tilde{x}^{*}))}}{C_*} \nonumber\\
&\leq \sqrt{\frac{2\max\{F(x^{k})-F(\tilde{x}^{*}),0\}}{C_*}}+\frac{2\|\Psi^tR/m\|_{\infty}}{C_*}.
\end{align}
On the other hand,
\begin{align}\label{eq002}
&F(x^{k})-F(\tilde{x}^{*})\|_{\infty}\leq \zeta^k (F(x^{0})-F(\tilde{x}^{*}))\nonumber\\
&\leq\zeta^k(\langle -\Psi^tR/m, x^0-\tilde{x}^*\rangle + \frac{C^*}{2} \|x^0-\tilde{x}^*\|_1\|x^0-\tilde{x}^*\|_{\infty})\nonumber\\
&\leq\zeta^k(\|\Psi^tR/m\|_{\infty}\|x^0-\tilde{x}^*\|_1+\frac{C^*}{2}\|x^0-\tilde{x}^*\|_1\|x^0-\tilde{x}^*\|_{\infty})\nonumber\\
&\leq \zeta^k(\|\Psi^tR/m\|_{\infty}|c|\|x^*\|_1+\frac{C^*}{2}c^2\|x^*\|_1^2)
\end{align}
where the first inequality uses  Lemma \ref{L9.2} and the second inequality uses the definition of $C^*$ in \eqref{CC} and Taylor expansion, the third  inequality follows from Cauchy Schwartz inequality, and the last one uses  $x^0 = \mathbf{0}$.
Combing the \eqref{apeq1} and \eqref{eq002} we get
\begin{align*}
\|x^{k}-\tilde{x}^{*}\|_{\infty} &\leq \zeta^{k/2}\sqrt{(2\|\Psi^tR/m\|_{\infty}|c|\|x^*\|_1/C_*+\frac{C^*}{C_*}c^2\|x^*\|_1^2)}+\frac{2\|\Psi^tR/m\|_{\infty}}{C_*}\\
&\leq \zeta^{k/2}(\sqrt{2\|\Psi^tR/m\|_{\infty}|c|\|x^*\|_1/C_*}+|c|\|x^*\|_1\sqrt{\frac{C^*}{C_*})})+\frac{2\|\Psi^tR/m\|_{\infty}}{C_*}.
\end{align*}
We  completes the proof by diving $|c|$.
\end{proof}
\section{Proof of Proposition \ref{recsupp}}\label{app:sup}
\begin{proof}
By \eqref{betaerror} in Theorem \ref{th:est} and the assumption on $|x^*|_{\mathrm{min}}$, we get $\mathrm{supp}(x^*) \subseteq A^k$.
Hence,  $\mathrm{supp}(x^*) = \subseteq A^k$  since we assume $\mathrm{supp}(x^*)=s$.  Then  $\mathrm{supp}(x^*)  =A^k = A^{k+1}$ as long as
$k\geq\log_{1/\zeta}(\frac{sm}{\log n} \frac{C^*C_1|c|^2}{16C_*(1+|c|C_3)^2\gamma_{\mathrm{min}}(\Sigma)}).$
This completes the proof.
\end{proof}

\bibliographystyle{abbrv}

\end{document}